\documentclass{IEEEtran}
\usepackage{amsfonts}
\usepackage{mathrsfs}
\usepackage{mathtools}
\usepackage{amssymb,amsmath}

\usepackage{algorithm}
\usepackage{algorithmic}
\usepackage{amsmath,amssymb,epsfig,graphics,subfigure}
\usepackage{theorem}
\usepackage{array,color}
\usepackage[compress]{cite}
\usepackage{tabularx}
\usepackage{multicol}
\usepackage{bm}
\usepackage{cite}
\usepackage{stfloats}
\usepackage{balance}
\usepackage{amsmath,amssymb,epsfig,graphics,subfigure}

\newtheorem{remark}{Remark}
\newtheorem{lemma}{Lemma}

\newtheorem{proof}{Proof}
\newtheorem{definition}{Definition}

\allowdisplaybreaks

\begin{document}

\title{Matrix-Monotonic Optimization $-$ Part I: Single-Variable Optimization}

\author{Chengwen Xing, \textsl{Member}, \textsl{IEEE}, Shuai Wang, \textsl{Member}, \textsl{IEEE},
 Sheng Chen, \textsl{Fellow}, \textsl{IEEE}, Shaodan Ma, \textsl{Member}, \textsl{IEEE}, H. Vincent Poor, \textsl{Fellow}, \textsl{IEEE}, and  Lajos Hanzo, \textsl{Fellow}, \textsl{IEEE}
 \thanks{This work was supported in part by the U.S. National Science Foundation under Grants CCF-0939370 and CCF-1513915.}
 \thanks{L. Hanzo would like to thank the ERC for his Advanced Fellow Award and the EPSRC for their financial support.}
\thanks{C.~Xing and S.~Wang are with School of Information and Electronics, Beijing
 Institute of Technology, Beijing 100081, China (E-mails: xingchengwen@gmail.com
 swang@bit.edu.cn)} %
\thanks{S.~Chen and L.~Hanzo are with School of Electronics and Computer Science,
 University of Southampton, U.K. (E-mails: sqc@ecs.soton.ac.uk, lh@ecs.soton.ac.uk).
 S. Chen is also with King Abdulaziz University, Jeddah, Saudi Arabia.} %
\thanks{S.~Ma is with Department of Electrical and Computer Engineering, University
 of Macau, Macao (E-mail: shaodanma@umac.mo).} %
\thanks{H.~V.~Poor is with Department of Electrical Engineering, Princeton University,
 Princeton, NJ 08544 USA (E-mail: poor@princeton.edu).} %
\vspace*{-5mm}
}

\maketitle

\begin{abstract}
 Matrix-monotonic optimization exploits the monotonic nature of positive semi-definite matrices to derive optimal diagonalizable structures for the matrix variables of matrix-variable optimization problems. Based on the optimal structures derived, the associated optimization  problems can be substantially simplified and underlying physical insights can also be revealed. In our work, a comprehensive framework of the applications of matrix-monotonic optimization to multiple-input multiple-output (MIMO) transceiver design is provided for a series of specific performance metrics under various linear constraints. This framework consists of two parts, i.e., Part-I for single-variable optimization and Part-II for multi-variable optimization. In this paper, single-variable matrix-monotonic optimization is investigated under various power constraints and various types of channel state information (CSI) condition. Specifically, three cases are investigated: 1)~both the transmitter and receiver have imperfect CSI; 2)~ perfect
 CSI is available at the receiver but the transmitter has no CSI; 3)~perfect CSI is
 available at the receiver but the channel estimation error at the transmitter is norm-bounded.
 In all three cases, the matrix-monotonic optimization framework can be used for
 deriving the optimal structures of the optimal matrix variables.
\end{abstract}


\section{Motivations}\label{sect:intro}

Antenna arrays are widely employed for improving the
 bandwidth- and/or the power-efficiency, resulting in the concept of
 multiple-input multiple-output (MIMO) system.
 \cite{JYang1994,Alamouti1998,Sampth01,Sampth03,
 Telatar1999,Scaglione2002,Palomar03,Feiten2007,Yadav2014}.
 Transceiver optimization is of critical importance for fulfilling the potential of
 MIMO communication systems \cite{Palomar03,Feiten2007,Yadav2014,Majorization}. MIMO transceiver
 optimization hinges on numerous factors, including their
 implementation issues, the availability of
 channel state information (CSI) and their system architectures. More specifically, MIMO transceivers can be classified into
 linear transceivers \cite{Yadav2014,Palomar03} and nonlinear transceivers
 \cite{Jiang2005,Weng2010TSP1,Liu2013TSP}. According to the different levels of CSI knowledge,
 MIMO transceiver designs can be classified into designs relying on perfect CSI
 \cite{Sampth01,Sampth03,Scaglione2002} and designs having partial
 CSI \cite{Zhang2008,JHWang2013,Ding09,Jafar2004,Jafar2005}. Finally, according to
 the system architecture, transceiver optimization can be used for point-to-point systems \cite{Majorization,Pastore2012}, for multi-user (MU) MIMO systems \cite{WYu2007},
 for distributed MIMO systems \cite{JFang2013,XingIET}, and
 for cooperative MIMO systems \cite{XingTSP2013,XingJSAC2012}.

 In all the above-mentioned multiple antenna aided systems, the corresponding optimization variables become matrix variables
 \cite{XingTSP201501}. As a result, optimization relying on matrix variables plays an
 important role in MIMO systems \cite{Jorswieck07}.
 Optimization relying on matrix variables is generally very challenging and such problems are much
 more difficult to solve than their counterparts with vector variables or scalar
 variables, because matrix variable based optimization usually involves complex matrix
 operations, such as the calculation of the determinants, inverses, matrix decompositions and so on.
 Furthermore, because of spatial multiplexing gains, MIMO systems are capable of supporting
 multiple data streams. This fact makes transceiver optimization problems inherently
 multi-objective optimization problems. For example, given  a limited transmit power,
 any specific transceiver optimization is a tradeoff between the performance of
 different data streams. This is the reason why there exists a rich body of work addressing
 various different MIMO transceiver designs \cite{Palomar03,Majorization}.

 Any transceiver optimization problem hinges on the  fundamental elements of the objective
 function and the specific optimization tools used for finding the extremities of the objective function. The more components the objective function has, the larger the search space becomes, which often makes a full hard utilization. A third related component is constituted by the constraints. The most widely used objective functions or performance metrics
  of MIMO transceiver optimization include the classic mean square error (MSE) minimization,
 signal to interference plus noise ratio (SINR) maximization or mutual information maximization,
 bit error rate (BER) minimization, etc, \cite{Palomar03}. Different performance metrics
 reflect different design preferences and different tradeoffs among the transmitted
 data streams \cite{Majorization}. Transceiver optimization problems using different performance metrics imposes different degrees of difficulty to solve. Furthermore, different objective functions also correspond
 to different implementation strategies resulting in, for example, linear transceivers, nonlinear transceivers using Tomlinson-Harashima precoding (THP) or decision feedback equalizer (DFE) etc.
 \cite{XingTSP201501,Weng2010TSP1,Weng2010TSP2,Liu2013TSP}. Suffice to say that the specific choice of the objective function has a more substantial impact on the overall MIMO design than that of the tools used for optimizing it.

 On the other hand, there are many different types of power constraints, such as the sum power
 constraint \cite{XingTSP201501}, per-antenna power constraint
 \cite{Mai2011,Christopoulos,Toli2008,Wu2016,Luo2008,Dong2013,XingTSP201601}, shaping constraint \cite{XingTSP201502,Palomar2004}, joint power constraints \cite{Dai2012}, cognitive
 constraint \cite{XingTSP201601}, etc. The most
 widely used power constraint is the sum power constraint requiring the sum
 of the powers at all the transmit antennas to be lower than a threshold. In communication
 systems, usually each antenna has its own amplifier \cite{WYu2007}. Therefore, the
 per-antenna power constraint is more practical than the sum power constraint. However,
 the per-antenna power constraint is more challenging to consider than the sum power
 constraint \cite{XingTSP201601,Christopoulos,Toli2008,Dong2013,WYu2007}. The existing literature has revealed that if different
 transmit antennas have the same statistics, the performance gain of considering the more
 challenging per-antenna power constraint based design over using the simpler sum power
 constraint design is negligible \cite{Luo2008}. Thus, under the scenario of similar statistics for different
 transmit antennas, the sum power constraint is an effective modeling technique. It is worth noting
 however that in some cases, as in distributed antenna systems or heterogeneous networks,
 different antennas have significantly different statistics, and thus the per-antenna
 power constraint cannot be replaced by the sum power constraint without a significant
 performance loss \cite{Luo2008,XingTSP201601}. Moreover, considering other practical constraints, such as signal
 variances or the peak-to-average-ratio, joint power constraints or other types of constraints
 have to be taken into account \cite{Dai2012}.

 It can be readily seen from the existing literature \cite{Majorization,XingTSP201501,XingTSP201601} that the underlying design principles
 for various transceiver optimization problems are almost the same. Generally, the main idea is
  taking advantage of the specific structure of the underlying optimization problem
 to simplify the transceiver optimization. Optimization theory plays an important role in
 MIMO transceiver optimization, and in the past decade many elegant results have been
 derived based on convex optimization theory \cite{XingIET,Jorswieck07}. Deriving
 optimal structures is critical in transceiver optimization
 \cite{Scaglione2002,Palomar03,JYang1994}. Clearly, a general-purpose
 optimal structure that can cover every MIMO transceiver optimization problems does not exist, and most the research has been focused on finding an optimal diagonalizable structure for MIMO transceiver
 optimization. This is because based on the optimal diagonalizable structures of the MIMO
 transceivers, the corresponding optimization problems can be substantially simplified and deeply
 underlying physical insights can also be revealed \cite{Scaglione2002,Palomar03,JYang1994}.

 Again, optimization variables of MIMO transceiver designs are generally matrix variables.
 Matrix-monotonic optimization exploits the monotonic nature of positive semi-definite matrices
 to derive optimal structures of the matrix variables in the underlying optimization
 problems \cite{XingTSP201501,XingTSP201502,XingTSP201601}. Based on matrix-monotonic
 optimization, the matrix variables can be substantially simplified into vector variables. The optimal
 structures delivered by matrix-monotonic optimization, therefore, greatly simplify
 complicated MIMO transceiver designs and make the underlying physical interpretation more
 transparent. From a matrix-monotonic optimization perspective, MIMO transceiver
 optimization problems relying on different objective functions and power constraints can be unified
 and, therefore, their associated optimal structures can be derived using the same
 matrix-monotonic optimization tool. Exploiting matrix-monotonic optimization
 is a powerful mathematical tool conceived for solving challenging matrix-variable
 transceiver optimization problems.

 This paper offers a comprehensive and novel work for matrix-monotonic optimization for a series of specific performance metrics under linear constraints in the context of MIMO
 transceiver optimization. Matrix monotonic optimization problem with various levels of CSI
 is investigated in depth.  Our main contributions are listed as follows.
\begin{itemize}
\item In contrast to~\cite{XingTSP201501} with only simple sum power constraint,
  the framework of matrix-monotonic
    optimization investigated in this treatise is subjected to diverse
    power constraints, including the sum power constraint, multiple
    weighted power constraints, joint power constraints and shaping
    constraints. In other words, the framework investigated in this
    paper subsumes the solutions in~\cite{XingTSP201501} and several
    other MIMO transceiver optimization solutions as its special cases.

 \item In contrast to~\cite{Palomar03} and \cite{Jiang2005}, where the
   linear and nonlinear transceiver designs are investigated
   separately under only the sum power constraint, the framework
   proposed in this paper unifies the families of linear and nonlinear
   MIMO transceiver optimization under the sum power constraint,
   shaping constraint, joint power constraints and multiple weighted
   power constraints.

\item Moreover, robust MIMO transceiver optimization relying on
  partial CSI under various power constraints is investigated based on
  the matrix-monotonic optimization framework. Specifically, the
  following three cases are investigated:
 \begin{enumerate}
 \item[1)] Both the transmitter and receiver have only imperfect CSI,
 \item[2)] The receiver has perfect CSI but the transmitter has only channel statistics,
 \item[3)] The receiver has perfect CSI but the channel estimate available at the transmitter is subject to a certain uncertainty norm-bounded error.
 \end{enumerate}
 Although having imperfect CSI makes the MIMO transceiver optimization more complex and
 challenging, the proposed matrix-monotonic optimization framework is still capable of
 deriving the underlying optimal structures.
\end{itemize}

 The remainder of this paper is organized as follows. In Section~\ref{S2}, we present the
 fundamentals of the matrix-monotonic optimization framework. Then Section~\ref{S3}
 investigates classic Bayesian robust matrix-monotonic optimization for robust transceiver
 design when the channel estimation errors are Gaussian distributed. In Section~\ref{S4},
 stochastic robust matrix-monotonic optimization is investigated for MIMO transceiver
 optimization where the receiver has perfect CSI but the transmitter knows only the channel
 statistics. Section~\ref{S5} is devoted to  worst case matrix-monotonic optimization,
 which focuses on transceiver optimization in the face of  norm-bounded channel estimation errors.

\noindent $\textbf{Notation:}$ The following notational conventions are adopted throughout our discussions. The normal-faced
 letters denote scalars, while bold-faced lower-case and upper-case letters denote vectors
 and matrices, respectively. $\bm{Z}^{\rm H}$, ${\rm Tr}(\bm{Z})$ and $|\bm{Z}|$ denote
 the Hermitian transpose, trace and determinant of complex matrix $\bm{Z}$, respectively.
 Statistical expectation is denoted by $\mathbb{E}\{\cdot \}$, and $a^+=
 \max\{0,a\}$, while $(\cdot )^{\rm T}$ denotes the vector/matrix transpose operator.
 $\bm{Z}^{\frac{1}{2}}$ is the Hermitian square root of $\bm{Z}$ which is positive
 semi-definite. The $i$th largest eigenvalue of $\bm{Z}$ is denoted by $\lambda_i(\bm{Z})$,
 and the $i$th-row and $j$th-column element of $\bm{Z}$ is denoted by $[\bm{Z}]_{i,j}$,
 while $\bm{d}[\bm{Z}]$ denotes the vector consisting of the diagonal elements of $\bm{Z}$
 and ${\rm diag}\big\{\{\bm{A}_k\}_{k=1}^K\big\}$ denotes the block diagonal matrix
 whose diagonal sub-matrices are $\bm{A}_1,\cdots ,\bm{A}_K$. The symbol $\bm{d}^2[\bm{Z}]$ denotes the vector consisting of the squared moduli of  the diagonal elements of $\bm{Z}$.
 Additionally, the $i$th element of a vector $\bm{z}$ is denoted by $[\bm{z}]_i$. The
 identity matrix of appropriate dimension is denoted by $\bm{I}$, and $\otimes$ is the
 Kronecker product. In this paper, $\bm{\Lambda}$ always denotes a diagonal matrix, and the
 expressions $\bm{\Lambda}\searrow$ and $\bm{\Lambda}\nearrow $ represent a rectangular or
 square diagonal matrix with the diagonal elements in descending order and ascending order,
 respectively.

\section{Fundamentals of Matrix-Monotonic Optimization}\label{S2}

An optimization problem with a real-valued objective function $f_0( \cdot )$ that depends on a complex matrix variable $\bm{X}$ is generally
 formulated as
\begin{align}\label{general_optimization} 
 \begin{array}{cl}
 \min\limits_{\bm{X}\in {\mathcal C}} &  f_0(\bm{X}) , \\
 \rm{s.t.} & \psi_i(\bm{X})\le 0 , 1\le i\le I ,
 \end{array}
\end{align}
 where $\psi_i(\cdot )$,
 $1\le i\le I$, are the constraint functions and ${\mathcal{C}}$ denotes the complex matrix set. A wide range of optimization problems
 can be cast in this optimization framework, including the classic MIMO transceiver
 optimization \cite{Majorization}, training designs \cite{XingTSP201501}, MIMO radar waveform optimization \cite{XingTSP201501}, etc. In order
 to analyze the properties of this generic optimization problem, we first discuss
 two of its basic components, namely, the objective function and the constraints,
 separately.

 \renewcommand{\arraystretch}{1.5}
 \begin{table}[tp!]
\vspace*{-2mm}
\caption {The objective functions}
\vspace*{-5mm}
\begin{center}
\scalebox{1.2}{
\begin{tabular}{||l|l||}
\hline
 Index & Objective function  \\ \hline
\hspace*{-2mm} \textbf{Obj.\,1} & \hspace*{-3mm} $-\log\big|{\bm{X}}^{\rm H}\bm{\Pi}{\bm{X}}
  +\bm{\Phi}\big|$  \\ \hline
\hspace*{-2mm} \textbf{Obj.\,2} &\hspace*{-3mm} ${\rm Tr}\left(\left({\bm{X}}^{\rm H}\bm{\Pi}{\bm{X}}+\bm{\Phi}\right)^{-1}\right)$
  \\ \hline
\hspace*{-2mm} \textbf{Obj.\,3} & \hspace*{-3mm} ${\rm Tr}\left(\bm{A}^{\rm H}\left({\bm{X}}^{\rm H}\bm{\Pi}{\bm{X}}+\alpha\bm{I}\right)^{-1}\bm{A}\right)$  \\ \hline
\hspace*{-2mm} \textbf{Obj.\,4} &\hspace*{-3mm} $\log\big|\bm{A}^{\rm H}\left({\bm{X}}^{\rm H}\bm{\Pi}
  {\bm{X}}+\alpha\bm{I}\right)^{-1}\bm{A}+\bm{\Phi}\big|$  \\ \hline
\hspace*{-2mm} \textbf{Obj.\,5.1} & \hspace*{-3mm} $f_{\text{A-Schur}}^{\text{Convex}} \left(\bm{d}\big[\left({\bm{X}}^{\rm H}
  \bm{\Pi}{\bm{X}}+\alpha\bm{I}\right)^{-1}\big]\right)$  \\ \hline
\hspace*{-2mm} \textbf{Obj.\,5.2} & \hspace*{-3mm} $f_{\text{A-Schur}}^{\text{Concave}}\left(\bm{d}\big[\left({\bm{X}}^{\rm H}
  \bm{\Pi}{\bm{X}}+\alpha\bm{I}\right)^{-1}\big]\right)$ \\ \hline
\hspace*{-2mm} \textbf{Obj.\,6.1} & \hspace*{-3mm} $f_{\text{M-Schur}}^{\text{Convex}} \left(\bm{d}^2[\bm{L}]\right)\hspace*{-1mm},
  \left({\bm{X}}^{\rm H}\bm{\Pi}{\bm{X}}+\alpha\bm{I}\right)^{-1}\hspace*{-1mm}
  =\hspace*{-1mm}\bm{L}\bm{L}^{\rm H}$ \\ \hline
\hspace*{-2mm} \textbf{Obj.\,6.2} & \hspace*{-3mm} $f_{\text{M-Schur}}^{\text{Concave}}\left(\bm{d}^2[\bm{L}]\right)\hspace*{-1mm},
  \left({\bm{X}}^{\rm H}\bm{\Pi}{\bm{X}}+\alpha\bm{I}\right)^{-1}\hspace*{-1mm}
  =\hspace*{-1mm}\bm{L}\bm{L}^{\rm H}$ \\ \hline
\hspace*{-2mm} \textbf{Obj.\,7} & \hspace*{-3mm} $-\log\big|\bm{A}^{\rm H}{\bm{X}}^{\rm H}\bm{\Pi}
  {\bm{X}}\bm{A}+\bm{\Phi}\big|$ \\ \hline
\hspace*{-2mm} \textbf{Obj.\,8} &\hspace*{-3mm} ${\rm Tr}\left(\left(\bm{A}^{\rm H}{\bm{X}}^{\rm H}\bm{\Pi}
  {\bm{X}}\bm{A}+\alpha\bm{I}\right)^{-1}\right)$  \\ \hline
\hspace*{-2mm} \textbf{Obj.\,9} & \hspace*{-3mm} ${\rm Tr}\left(\bm{A}^{\rm H}\left({\bm{X}}^{\rm H}\bm{\Pi}
  {\bm{X}}+\bm{\Phi}\right)^{-1}\bm{A}\right)$  \\ \hline
\hspace*{-2mm} \textbf{Obj.\,10} & \hspace*{-3mm} $-\log\big|\bm{\Phi}\otimes \bm{\Sigma}_1+\big({\bm{X}}^{\rm H}
  \bm{\Pi}{\bm{X}}\big)\otimes \bm{\Sigma}_3\big|$ \\ \hline
\hspace*{-2mm} \textbf{Obj.\,11} & \hspace*{-3mm} $-\log\big| \bm{\Sigma}_1\otimes \bm{\Phi}+\bm{\Sigma}_2\otimes \big(
  {\bm{X}}^{\rm H}\bm{\Pi}{\bm{X}}\big) \big|$  \\ \hline
\hspace*{-2mm} \textbf{Obj.\,12} & \hspace*{-3mm} ${\rm Tr}\left(\left(\bm{\Phi}\otimes \bm{\Sigma}_1+\big({\bm{X}}^{\rm H}\bm{\Pi}{\bm{X}}\big)\otimes \bm{\Sigma}_2\right)^{-1} \right)$ \\ \hline
\hspace*{-2mm} \textbf{Obj.\,13} & \hspace*{-3mm} ${\rm Tr}\left(\left(\bm{\Sigma}_1\otimes \bm{\Phi}+\bm{\Sigma}_2\otimes
  \big({\bm{X}}^{\rm H}\bm{\Pi}{\bm{X}}\big) \right)^{-1} \right)$ \\ \hline
\hspace*{-2mm} \textbf{Obj.\,14} &  \hspace*{-3mm} ${\rm Tr}\left(\hspace*{-1mm}\big(\bm{A}\bm{A}^{\rm H}\big)\hspace*{-1mm}\otimes\hspace*{-1mm} \bm{\Sigma}_1\left(\bm{I}\hspace*{-1mm}
  +\hspace*{-1mm}\big({\bm{X}}^{\rm H}\bm{\Pi}{\bm{X}}\big)\otimes
  \bm{\Sigma}_2 \right)^{-1}\right)$  \\ \hline
\hspace*{-2mm} \textbf{Obj.\,15} & \hspace*{-3mm} ${\rm Tr}\left(\hspace*{-1mm}\bm{\Sigma}_1\hspace*{-1mm}\otimes \hspace*{-1mm} \big(\bm{A}\bm{A}^{\rm H}\big)\left(\bm{I}\hspace*{-1mm}
  +\hspace*{-1mm}\bm{\Sigma}_2\otimes \big({\bm{X}}^{\rm H}\bm{\Pi}{\bm{X}}\big)
  \right)^{-1} \right)$ \\ \hline
 \end{tabular}
}
\end{center}
\label{tab0}
\vspace*{-7mm}
\end{table}

\renewcommand{\arraystretch}{1.5}
\begin{table*}[tp!]
\vspace*{-2mm}
\caption {The objective functions and the optimal unitary matrices $\bm{Q}_{\bm{X}}$}
\vspace*{-5mm}
\begin{center}
\scalebox{1.2}{
\begin{tabular}{||l|l|l||}
\hline
 Index & Objective function & Optimum $\bm{Q}_{\bm{X}}$ \\ \hline\hline
 \textbf{Obj.\,1} & $-\log\big|\bm{Q}_{\bm{X}}^{\rm H}\bm{F}^{\rm H}\bm{\Pi}\bm{F}\bm{Q}_{\bm{X}}
  +\bm{\Phi}\big|$ & $\bm{U}_{\bm{F}\bm{\Pi}\bm{F}} \bar{\bm{U}}_{\bm{\Phi}}^{\rm H}$ \\ \hline
 \textbf{Obj.\,2} & ${\rm Tr}\left(\left(\bm{Q}_{\bm{X}}^{\rm H}\bm{F}^{\rm H}\bm{\Pi}\bm{F}
  \bm{Q}_{\bm{X}}+\bm{\Phi}\right)^{-1}\right)$ & $\bm{U}_{\bm{F}\bm{\Pi}\bm{F}}\bar{\bm{U}}_{\bm{\Phi}}^{\rm H}$
  \\ \hline
 \textbf{Obj.\,3} & ${\rm Tr}\left(\bm{A}^{\rm H}\left(\bm{Q}_{\bm{X}}^{\rm H}\bm{F}^{\rm H}\bm{\Pi}\bm{F}
  \bm{Q}_{\bm{X}}+\alpha\bm{I}\right)^{-1}\bm{A}\right)$ & $\bm{U}_{\bm{F}\bm{\Pi}\bm{F}}
  \bm{U}_{\bm{A}}^{\rm H}$ \\ \hline
 \textbf{Obj.\,4} & $\log\big|\bm{A}^{\rm H}\left(\bm{Q}_{\bm{X}}^{\rm H}\bm{F}^{\rm H}\bm{\Pi}\bm{F}
  \bm{Q}_{\bm{X}}+\alpha\bm{I}\right)^{-1}\bm{A}+\bm{\Phi}\big|$ & $\bm{U}_{\bm{F}\bm{\Pi}\bm{F}}
  \bm{U}_{\bm{A}\bm{\Phi}\bm{A}}^{\rm H}$ \\ \hline
 \textbf{Obj.\,5.1} & $f_{\text{A-Schur}}^{\text{Convex}} \left(\bm{d}\big[\left(\bm{Q}_{\bm{X}}^{\rm H}
  \bm{F}^{\rm H}\bm{\Pi}\bm{F}\bm{Q}_{\bm{X}}+\alpha\bm{I}\right)^{-1}\big]\right)$ &
  $\bm{U}_{\bm{F}\bm{\Pi}\bm{F}}\bm{U}_{\text{DFT}}^{\rm H}$ \\ \hline
 \textbf{Obj.\,5.2} & $f_{\text{A-Schur}}^{\text{Concave}}\left(\bm{d}\big[\left(\bm{Q}_{\bm{X}}^{\rm H}
  \bm{F}^{\rm H}\bm{\Pi}\bm{F}\bm{Q}_{\bm{X}}+\alpha\bm{I}\right)^{-1}\big]\right)$ &
  $\bm{U}_{\bm{F}\bm{\Pi}\bm{F}}$ \\ \hline
 \textbf{Obj.\,6.1} & $f_{\text{M-Schur}}^{\text{Convex}} \left(\bm{d}^2[\bm{L}]\right) \ {\rm with} \
  \left(\bm{Q}_{\bm{X}}^{\rm H}\bm{F}^{\rm H}\bm{\Pi}\bm{F}\bm{Q}_{\bm{X}}+\alpha\bm{I}\right)^{-1}
  =\bm{L}\bm{L}^{\rm H}$ & $\bm{U}_{\bm{F}\bm{\Pi}\bm{F}}\bm{U}_{\text{GMD}}^{\rm H}$ \\ \hline
 \textbf{Obj.\,6.2} & $f_{\text{M-Schur}}^{\text{Concave}}\left(\bm{d}^2[\bm{L}]\right) \ {\rm with} \
  \left(\bm{Q}_{\bm{X}}^{\rm H}\bm{F}^{\rm H}\bm{\Pi}\bm{F}\bm{Q}_{\bm{X}}+\alpha\bm{I}\right)^{-1}
  =\bm{L}\bm{L}^{\rm H}$ & $\bm{U}_{\bm{F}\bm{\Pi}\bm{F}}$ \\ \hline
 \textbf{Obj.\,7} & $-\log\big|\bm{A}^{\rm H}\bm{Q}_{\bm{X}}^{\rm H}\bm{F}^{\rm H}\bm{\Pi}\bm{F}
  \bm{Q}_{\bm{X}}\bm{A}+\bm{\Phi}\big|$ & $\bm{U}_{\bm{F}\bm{\Pi}\bm{F}}\bm{U}_{\bm{A}\bm{\Phi}\bm{A}}^{\rm H}$ \\ \hline
 \textbf{Obj.\,8} & ${\rm Tr}\left(\left(\bm{A}^{\rm H}\bm{Q}_{\bm{X}}^{\rm H}\bm{F}^{\rm H}\bm{\Pi}
  \bm{F}\bm{Q}_{\bm{X}}\bm{A}+\alpha\bm{I}\right)^{-1}\right)$ & $\bm{U}_{\bm{F}\bm{\Pi}\bm{F}}
  \bm{U}_{\bm{A}}^{\rm H} \ \text{(High SNR)}$ \\ \hline
 \textbf{Obj.\,9} & ${\rm Tr}\left(\bm{A}^{\rm H}\left(\bm{Q}_{\bm{X}}^{\rm H}\bm{F}^{\rm H}\bm{\Pi}
  \bm{F}\bm{Q}_{\bm{X}}+\bm{\Phi}\right)^{-1}\bm{A}\right)$ & $\bm{U}_{\bm{F}\bm{\Pi}\bm{F}}
  \bm{U}_{\bm{A}}^{\rm H} \ \text{(High SNR)}$ \\ \hline
 \textbf{Obj.\,10} & $-\log\big|\bm{\Phi}\otimes \bm{\Sigma}_1+\big(\bm{Q}_{\bm{X}}^{\rm H}
  \bm{F}^{\rm H}\bm{\Pi}\bm{F}\bm{Q}_{\bm{X}}\big)\otimes \bm{\Sigma}_2\big|$ &
  $\bm{U}_{\bm{F}\bm{\Pi}\bm{F}}\bar{\bm{U}}_{\bm{\Phi}}^{\rm H}$ \\ \hline
 \textbf{Obj.\,11} & $-\log\big| \bm{\Sigma}_1\otimes \bm{\Phi}+\bm{\Sigma}_2\otimes \big(
  \bm{Q}_{\bm{X}}^{\rm H}\bm{F}^{\rm H}\bm{\Pi}\bm{F}\bm{Q}_{\bm{X}}\big) \big|$ &
  $\bm{U}_{\bm{F}\bm{\Pi}\bm{F}}\bar{\bm{U}}_{\bm{\Phi}}^{\rm H}$ \\ \hline
 \textbf{Obj.\,12} & ${\rm Tr}\left(\left(\bm{\Phi}\otimes \bm{\Sigma}_1+\big(\bm{Q}_{\bm{X}}^{\rm H}
  \bm{F}^{\rm H}\bm{\Pi}\bm{F}\bm{Q}_{\bm{X}}\big)\otimes \bm{\Sigma}_2\right)^{-1} \right)$
  & $\bm{U}_{\bm{F}\bm{\Pi}\bm{F}}\bar{\bm{U}}_{\bm{\Phi}}^{\rm H}$ \\ \hline
 \textbf{Obj.\,13} & ${\rm Tr}\left(\left(\bm{\Sigma}_1\otimes \bm{\Phi}+\bm{\Sigma}_2\otimes
  \big(\bm{Q}_{\bm{X}}^{\rm H}\bm{F}^{\rm H}\bm{\Pi}\bm{F}\bm{Q}_{\bm{X}}\big) \right)^{-1} \right)$
  & $\bm{U}_{\bm{F}\bm{\Pi}\bm{F}}\bar{\bm{U}}_{\bm{\Phi}}^{\rm H}$ \\ \hline
 \textbf{Obj.\,14} & ${\rm Tr}\left(\big(\bm{A}\bm{A}^{\rm H}\big)\otimes \bm{\Sigma}_1\left(\bm{I}
  +\big(\bm{Q}_{\bm{X}}^{\rm H}\bm{F}^{\rm H}\bm{\Pi}\bm{F}\bm{Q}_{\bm{X}}\big)\otimes
  \bm{\Sigma}_2 \right)^{-1} \right)$ & $\bm{U}_{\bm{F}\bm{\Pi}\bm{F}}\bm{U}_{\bm{A}}^{\rm H}$ \\ \hline
 \textbf{Obj.\,15} & ${\rm Tr}\left(\bm{\Sigma}_1\otimes \big(\bm{A}\bm{A}^{\rm H}\big)\left(\bm{I}
  +\bm{\Sigma}_2\otimes \big(\bm{Q}_{\bm{X}}^{\rm H}\bm{F}^{\rm H}\bm{\Pi}\bm{F}\bm{Q}_{\bm{X}}\big)
  \right)^{-1} \right)$ & $\bm{U}_{\bm{F}\bm{\Pi}\bm{F}}\bm{U}_{\bm{A}}^{\rm H}$ \\ \hline
 \end{tabular}
}
\end{center}
\label{tab1}
\vspace*{-7mm}
\end{table*}

\subsection{Objective Functions}\label{S2.1}

 The objective function reflects the cost or utility of the optimization problem.
 In this paper, all the optimization problems discussed are formulated  with the
 objective of minimizing a cost function. Let us now discuss the commonly used
 objective functions, listed in  Table~\ref{tab0}. For transceiver optimization, the mutual information is one of the most important performance
 metrics. For training optimization, the mutual information is also an important
 performance metric as it reflects the correlation between the estimated parameters
 and the true parameters. In these cases, the objective function is given by
 \textbf{Obj.\,1} \cite{NewTSP}, where $\bm{\Pi}$ and $\bm{\Phi}$ are constant
 positive semi-definite matrices which have different physical meanings for different
 systems. The MSE is another important performance metric for transceiver or training
 optimization, which reflects how accurately a signal can be recovered rather than
 how much information can be transmitted. For the optimization problem of sum MSE
 minimization, the objective function is given in the form of \textbf{Obj.\,2} \cite{NewTSP}.

 Generally, the MSE formulation for linear transceiver optimization is determined
 by the specific signal model considered. For example, in a dual-hop AF
 MIMO relaying network, the MSE minimization has \textbf{Obj.\,3} \cite{XingCL}, where $\alpha$ is
 a positive scalar and $\bm{A}$ is a constant complex matrix. Similarly, the mutual information
 maximization for a dual-hop AF MIMO relaying network aims at minimizing the objective
 function \textbf{Obj.\,4} \cite{XingCL}\footnote{This conclusion is achieved is based on the fact that maximizing mutual information is equivalent to minimizing the determinant of the MSE matrix \cite{XingTSP201501}}. For linear transceiver optimization, to realize different
 levels of fairness between different transmitted data streams, a general objective
 function can be formulated as an additively Schur-convex function \cite{Palomar03} or additively
 Schur-concave function \cite{Palomar03} of the diagonal elements of the MSE matrix, which are given
 by \textbf{Obj.\,5.1} and \textbf{Obj.\,5.2} \cite{Majorization}, respectively.
 The additively
 Schur-convex function $f_{\text{A-Schur}}^{\rm convex}(\cdot )$ and the additively
 Schur-concave function $f_{\text{A-Schur}}^{\rm concave}(\cdot )$ represent different
 levels of fairness among the diagonal elements of the data MSE matrix. In addition,  $f_{\text{A-Schur}}^{\rm convex}(\cdot )$ and $f_{\text{A-Schur}}^{\rm concave}(\cdot )$ are both increasing functions with respect to the vector variables.

 When nonlinear transceivers are chosen for improving the BER performance at the cost
 of increased complexity, e.g., THP or DFE, the objective functions of the transceiver
 optimization can be formulated as a multiplicative Schur-convex function or a
 multiplicative Schur-concave function of the vector consisting of the squared
 diagonal elements of the Cholesky-decomposition triangular matrix of the MSE matrix,
 that is, \textbf{Obj.\,6.1} and \textbf{Obj.\,6.2} \cite{XingTSP201501}, respectively,  where $\bm{L}$ is
 a lower triangular matrix. The multiplicatively Schur-convex function
 $f_{\text{M-Schur}}^{\rm convex}(\cdot )$ and the multiplicatively Schur-concave function
 $f_{\text{M-Schur}}^{\rm concave}(\cdot )$ reflect the different levels of fairness
 among the different data streams, i.e., different tradeoffs among the performance of
 different data steams \cite{XingTSP201501}. In addition,  $f_{\text{M-Schur}}^{\rm convex}(\cdot )$ and $f_{\text{M-Schur}}^{\rm concave}(\cdot )$ are both increasing functions with respect to the vector variables.

 In wireless communication designs, even for the same system or the same optimization
 problem, the mathematical formulae are not unique. More specifically, for the mutual
 information maximization, we have the alternative objective function \textbf{Obj.\,7} \cite{XingTSP201501}.
 Similarly, the sum MSE minimization has the alternative objective function
 \textbf{Obj.\,8} \cite{XingTSP201501}. Moreover, the weighted MSE minimization can be considered as a
 general extension of the sum MSE minimization by introducing a weighting matrix, which
 has the objective function \textbf{Obj.\,9}.

 As discussed in the existing literature, some MIMO system optimization problems may
 involve Kronecker products due to ${\rm{vec}}(\cdot)$ operations \cite{XingTSP201501}. The optimization problems relying on
 Kronecker product usually look very complicated. In this paper, the pair of
 optimization problems relying on either the matrix determinant or on the the matrix trace are discussed that
 involve Kronecker products. Based on \textbf{Obj.\,1}, we have the extended Kronecker
 structured objective function \textbf{Obj.\,10}, which is equivalent to
 \textbf{Obj.\,11} \cite{XingTSP201501}. It can readily be seen that with the choice of $\bm{\Sigma}_1=\bm{\Sigma}_2$,
 \textbf{Obj.\,10} and \textbf{Obj.\,11} are equivalent to \textbf{Obj.\,1}. In this
 paper, we also consider a more general case in which $\bm{\Sigma}_1$ and
 $\bm{\Sigma}_2$ have the same eigenvalue decomposition (EVD) unitary matrix. Under
 this assumption and based on \textbf{Obj.\,2}, we have the extended Kronecker
 structured objective function \textbf{Obj.\,12}, which is equivalent to \textbf{Obj.\,13}.
 Similarly, based on \textbf{Obj.\,3}, we have the objective function \textbf{Obj.\,14},
 which is also equivalent to \textbf{Obj.\,15}. In our following discussions involving
 \textbf{Obj.\,10} to \textbf{Obj.\,15}, it is always assumed that $\bm{\Sigma}_1$ and
 $\bm{\Sigma}_2$ have the same EVD unitary matrix.

\subsection{Constraint Functions}\label{S2.2}

 In practical communication system designs, typically the associated optimization
 problems have constraints, and these constraints have different physical meanings
 for different communication systems.

 The most natural constraints are the power constraints, since practical amplifiers have
  certain maximum transmit power thresholds. The simplest power constraint, is the sum power
 constraint which can be expressed as
\begin{align}\label{eq2}
 \textbf{Constraint 1:} \ \  {\rm Tr}\big(\bm{X}\bm{X}^{\rm H}\big)\le P .
\end{align}
 With the sum power constraint, the optimization problems associated with training
 sequence designs or transceiver designs are subjected to the constraint of the power
 sum of all the transmit antennas. In practical systems, each antenna has its own
 power amplifier and, therefore, the per-antenna power constraints or individual power
 constraints provide a more reasonable power constraint model, which is expressed as
\begin{align}\label{eq3}
 \textbf{Constraint 2:} \ \  \big[\bm{X}\bm{X}^{\rm H}\big]_{n,n}\le P_n, \ \ n=1,\cdots ,N,
\end{align}
 where we have assumed that the number of transmit antennas is $N$ and the matrix variable
 $\bm{X}$ has $N$ rows. The per-antenna power constraint (\ref{eq3}) may be more practical
 but it does not include the sum power constraint (\ref{eq2}) as its special case.

 In sophisticated communication networks, the constraints are not limited to reflect the maximum
 power constraints at the transmit antennas for the desired signal but they also reflect many
 other constraints such as the interference constraints between adjacent links. A more
 general power constraint is the following one having multiple weighted components \cite{NewTSP}
\begin{align}\label{eq4}
 \textbf{Constraint 3:} \ \   {\rm Tr}\big(\bm{\Omega}_i\bm{X}\bm{X}^{\rm H}\big) \le P_i,
  \ \ i=1,\cdots ,I ,
\end{align}
 where $I$ is the number of weighted power constraints. \textbf{Constraint 3} is more
 general than \textbf{Constraint 1} and \textbf{Constraint 2}. The constraint model
 (\ref{eq4}) includes the sum power constraint (\ref{eq2}) and  per-antenna power constraint
 (\ref{eq3}) as its special cases. Specifically, by choosing $I=1$ and $\bm{\Omega}_1=\bm{I}$,
 this power constraint model becomes the sum power constraint (\ref{eq2}). Furthermore,
 when $I=N$ and $\bm{\Omega}_i$ is the matrix whose $i$th diagonal element is one and
 all the other elements are zeros, this model is exactly the per-antenna power constraint
 (\ref{eq3}).

 In order to avoid or control the interference, it is
  expected to be cast to the null space of the desired signals, hence
  the signal and interference become orthogonal to each other. In order
  to achieve this, constraints can be imposed on the covariance
  matrix of the transmitted signal, which are referred to as spectral
  mask constraints~\cite{Palomar2004}.
A classic example is the shaping constraint, which is formulated as the following matrix inequality \cite{Palomar2004,XingTSP201502}
\begin{align}\label{eq5}
 \textbf{Constraint 4:} \ \  \bm{X}\bm{X}^{\rm H} \preceq \bm{R}_{\rm s} .
\end{align}
 From matrix inequality theory, this constraint is equivalent to \cite[471]{Horn1990}
\begin{align}\label{eq6}
 {\rm Tr}\big(\bm{\Omega}_i\bm{X}\bm{X}^{\rm H}\big) \le {\rm Tr}\big(\bm{\Omega}_i\bm{R}_{\rm s}\big) ,
\end{align}
 for any positive semi-definite matrix $\bm{\Omega}_i$. Based on this fact, we can argue
 that the shaping constraint represents a special case of the multiple weighted power constraint.
 A simplified version of \textbf{Constraint 4} is the constraints imposed on the eigenvalues of
 the covariance matrix $\bm{X}\bm{X}^{\rm H}$ formulated  as
\begin{align}\label{eq7}
 \textbf{Constraint 5:} \ \  \lambda_i\big(\bm{X}\bm{X}^{\rm H}\big)\le \tau_i .
\end{align}
 A widely used eigenvalue constraint is the constraint on the maximum eigenvalue,
 $\lambda_1\big(\bm{X}\bm{X}^{\rm H}\big)\! \le\! \tau_1$, which is equivalent to \cite{XingTSP201502}
\begin{align}\label{eq8}
 \bm{X}\bm{X}^{\rm H} \preceq \tau_1 \bm{I} .
\end{align}
 This constraint can be used together with the sum power constraint to limit the
 transmitter's peak power. This is because most of the existing power constraints are
 based on statistical averages, while from a practical implementation perspective,
 the power constraint is an instantaneous constraint instead of being an average one \cite{Dai2012}. This kind of
 combined power constraint is termed as the joint power constraint, which is expressed as \cite{XingTSP201502}
\begin{align}\label{eq9}
 \textbf{Constraint 6:}  \  {\rm Tr}\big(\bm{X}\bm{X}^{\rm H}\big) \le P , \bm{X}\bm{X}^{\rm H} \preceq \tau_1 \bm{I} .
\end{align}

In cognitive radio communications, the interference imposed by the secondary user on the primary user must
 be smaller than a threshold and this constraint can be written in the following form
\begin{align}\label{eq10}
 \textbf{Constraint 7:} \ \   {\rm Tr}\big(\bm{H}_{\rm c}\bm{X}\bm{X}^{\rm H}
  \bm{H}_{\rm c}^{\rm H}\big) \le \tau_{\rm C} ,
\end{align}where $\bm{H}_{\rm c}$ is the channel matrix between the secondary user and primary user, while $\tau_{\rm C}$ is the interference threshold. This kind of constraint is also a
 special case of \textbf{Constraint 3}.

In summary, all the power constraint models discussed
  above represent the different physical constraints on the covariance
  matrix of the transmit signal, which equals
  $\bm{X}\bm{X}^{\rm{H}}$. These constraints shape the positive
  semidefinite covariance matrix.  For the simplest sum power
  constraint model, the sum of the eigenvalues of the covariance
  matrix has to be smaller than a threshold. For the multiple weighted
  power constraint model, the eigenvalues of the covariance matrices
  are constrained in the polyhedron region constructed by the multiple
  weighting matrices. In this case, except for the restrictions on the
  eigenvalues, the constraints also restrict the unitary matrix in the
  eigenvalue decomposition of the covariance matrix. Moreover, for the
  joint power constraint model the sum of the eigenvalues and the
  maximum eigenvalue are simultaneously smaller than the predefined
  thresholds. The upper-bound on the maximum eigenvalue significantly
  impacts the power allocations on the eigenchannels. For example,
  some subchannels that are allocated zero power for the sum
  power constraint will be assigned non-zero powers for the
  joint power constraints.

\begin{figure*}[bp!]\setcounter{equation}{18}
\vspace*{-5mm}
\hrule
\begin{align}
 \textbf{Matrix Inequality 1:} \ & \sum\nolimits_{i=1}^N \lambda_{i-1+N}(\bm{C}) \lambda_i(\bm{D})\le
  {\rm Tr}(\bm{C}\bm{D})\le \sum\nolimits_{i=1}^N \lambda_i(\bm{C}) \lambda_i(\bm{D}) , \label{eq19} \\
 \textbf{Matrix Inequality 2:} \ & \sum\nolimits_{i=1}^N \big(\lambda_{i-1+N}(\bm{C})+\lambda_i(\bm{D})\big)^{-1}
  \le {\rm Tr}\big( (\bm{C}+\bm{D})^{-1} \big) \le \sum\nolimits_{i=1}^N \big(\lambda_i(\bm{C})+
  \lambda_i(\bm{D})\big)^{-1} , \label{eq20} \\
 \textbf{Matrix Inequality 3:} \ & \prod\nolimits_{i=1}^N \big(\lambda_i(\bm{C}) + \lambda_i(\bm{D})\big)
  \le |\bm{C}+\bm{D}|\le \prod\nolimits_{i=1}^N \big(\lambda_{i-1+N}(\bm{C})+ \lambda_i(\bm{D})\big) ,
  \label{eq21} \\
 \textbf{Matrix Inequality 4:} \ & \prod\nolimits_{i=1}^N \big(\lambda_{i-1+N}(\bm{C})\lambda_i(\bm{D})+1\big)
  \le |\bm{C}\bm{D}+\bm{I}| \le \prod\nolimits_{i=1}^N \big(\lambda_i(\bm{C})\lambda_i(\bm{D})+1\big) .
  \label{eq22}
\end{align}
\vspace*{-1mm}
\end{figure*}

Before turning attention to discuss the optimization problem (\ref{general_optimization}), two
fundamental definitions are first introduced.\setcounter{equation}{10}
\begin{definition}\label{D1}
 A constraint $\psi (\bm{X})\le 0$ is a left unitary invariant constraint if we have
\begin{align}\label{eq11}
 \psi\big(\bm{Q}_{\rm L}\bm{X}\big) = \psi (\bm{X}) ,
\end{align}
 where $\bm{Q}_{\rm L}$ is an arbitrary unitary matrix.
\end{definition}

\begin{definition}\label{D2}
 A constraint $\psi (\bm{X})\le 0$ is a right unitarily-invariant constraint if we have
\begin{align}\label{eq12}
 \psi\big(\bm{X}\bm{Q}_{\rm R}\big) = \psi (\bm{X}) ,
\end{align}
 where $\bm{Q}_{\rm R}$ is an arbitrary unitary matrix.
\end{definition}

It is worth noting that all the constraints discussed above are right unitarily-invariant. Specifically, in \textbf{Constraints 1} to \textbf{7}, after replacing $\bm{X}$ by $\bm{X}\bm{Q}_{\rm R}$ it can be concluded that these constraints do not change.
Therefore, we can focus our attention on the family of right unitarily-invariant constraints only. In particular, we will focus our attention on the shaping constraint, joint power constraints and multiple weighted power constraints.

\subsection{Matrix-Monotonic Optimization}\label{S2.3}

 Based on the above discussions, with the objective functions in Table~\ref{tab0}, the generic optimization problem of MIMO systems
 can be formulated as
\begin{align}\label{eq13}
 \textbf{Opt.\,1.1:} & \min\limits_{\bm{X}}   f\big(\bm{X}^{\rm H}\bm{\Pi}\bm{X}\big) ,
   {\rm s.t.}  \psi_i(\bm{X})\le 0 , \ 1\le i\le I .
\end{align} The function $f(\cdot)$ is matrix monotone decreasing function \cite{Jorswieck07,BarrySimon,Xingzhi_Zhang}.
 Since the constraints are right unitarily-invariant, we introduce the auxiliary matrix
 variable $\bm{F}$ and express the original matrix variable $\bm{X}$ as
\begin{align}\label{eq14}
 \bm{X} = \bm{F} \bm{Q}_{\bm{X}} ,
\end{align}
 where $\bm{Q}_{\bm{X}}$ is an arbitrary unitary matrix. Based on (\ref{eq14}), the
 optimization problem (\ref{eq13}) can be reformulated as
\begin{align}\label{eq15}
\begin{array}{cl}
 \min\limits_{\bm{F},\bm{Q}_{\bm{X}}} &  f\big(\bm{Q}_{\bm{X}}^{\rm H}\bm{F}^{\rm H}\bm{\Pi}
  \bm{F}\bm{Q}_{\bm{X}}\big) , \\
 {\rm s.t.} & \psi_j\big(\bm{F}\bm{Q}_{\bm{X}}\big) = \psi_i(\bm{F}\big)\le 0 , \ 1\le i\le I,
\end{array}
\end{align}where the specific objective functions are given in the left column of Table~\ref{tab1}.
 Note that the constraints do not depend on $\bm{Q}_{\bm{X}}$. Therefore, the
 optimal $\bm{Q}_{\bm{X}}$ is independent of the constraints.

\subsubsection{Optimization of ${\bm{Q}}_{\bm{X}}$}

 Generally, there are two basic approaches to optimize $\bm{Q}_{\bm{X}}$. The first one
 is based on the basic matrix inequality and the other is based on majorization theory.

\underline{\textbf{Basic Matrix Inequalities}}
 Typically, the extreme values of basic matrix operations e.g., trace, determinant, etc.,
 are functions of the eigenvalues of the matrices involved. Given the positive semi-definite
 matrices $\bm{C}\in \mathbb{C}^{N\times N}$ and $\bm{D}\in \mathbb{C}^{N\times N}$, we
 consider the following EVDs
\begin{align}
 \bm{C} = \bm{U}_{\bm{C}}\bm{\Lambda}_{\bm{C}}\bm{U}_{\bm{C}}^{\rm H} \ \ \text{with} \ \
  \bm{\Lambda}_{\bm{C}} \searrow , \label{eq16} \\
 \bm{D} = \bm{U}_{\bm{D}}\bm{\Lambda}_{\bm{D}}\bm{U}_{\bm{D}}^{\rm H} \ \ \text{with} \ \
  \bm{\Lambda}_{\bm{D}} \searrow , \label{eq17} \\
 \bm{D} = \bar{\bm{U}}_{\bm{D}}\bar{\bm{\Lambda}}_{\bm{D}}\bar{\bm{U}}_{\bm{D}}^{\rm H} \ \
  \text{with} \ \ \bar{\bm{\Lambda}}_{\bm{D}} \nearrow , \label{eq18}
\end{align}
 where $\bm{\Lambda}_{\bm{D}}$ and $\bar{\bm{\Lambda}}_{\bm{D}}$ consist of the eigenvalues
 of $\bm{D}$ arranged in descending order and ascending order, while $\bm{U}_{\bm{D}}$ and
 $\bar{\bm{U}}_{\bm{D}}$ contain the corresponding eigenvectors of $\bm{D}$, respectively.
 Then we have the four basic matrix inequalities, ranging from (\ref{eq19}) to (\ref{eq22}), shown at the
 bottom of this page. Furthermore, in both \textbf{Matrix Inequality 1} \cite[P340, P341]{Marshall79} and
 \textbf{Matrix Inequality 2} \cite[Appendix A]{XingTSP201501}, the left equality holds when $\bm{U}_{\bm{C}}=
 \bar{\bm{U}}_{\bm{D}}$, and the right equality holds when $\bm{U}_{\bm{C}}=\bm{U}_{\bm{D}}$;
 while in both \textbf{Matrix Inequality 3} \cite[P333, P334]{Marshall79} and \textbf{Matrix Inequality 4}, the left
 equality holds when $\bm{U}_{\bm{C}}=\bm{U}_{\bm{D}}$, and the right equality holds when
 $\bm{U}_{\bm{C}}=\bar{\bm{U}}_{\bm{D}}$ \cite{XingTSP201501}.

 \underline{\textbf{Majorization Theory}}
 Majorization theory constitutes an important branch of matrix equality theory \cite{Jorswieck07,Marshall79}.
 We have the following two important definitions.

\begin{definition}[\!\!\cite{Marshall79}]\label{D3}
 For two vectors $\bm{x},\bm{y}\in \mathbb{R}^N$, $\bm{x}$ is said to be majorized by $\bm{y}$,
 denoted as $\bm{x}\prec \bm{y}$, when the following inequalities are satisfied:
 $\bigcirc_{i=1}^k [\bm{x}]_i\le \bigcirc_{i=1}^k [\bm{y}]_i$, for $1\le k\le N-1$, and
 $\bigcirc_{i=1}^N [\bm{x}]_i=\bigcirc_{i=1}^N [\bm{y}]_i$, where $\bigcirc$ denotes a
  mathematical operator.
\end{definition}

 In the following, we only consider the addition and product operators of
 $\bigcirc =\sum$ and  $\bigcirc =\prod$.

\begin{definition}[\!\!\cite{Marshall79}]\label{D4}
 A real-valued function $\phi :\mathbb{R}^N \to \mathbb{R}$ is additively or
 multiplicatively Schur-convex for any $\bm{x},\bm{y}$ in the feasible set, $\bm{x}
 \prec \bm{y} \rightarrow \phi (\bm{x})\le \phi (\bm{y})$. On the other hand, $\phi$ is
 additively or multiplicatively Schur-concave when $\bm{x} \prec \bm{y} \rightarrow
 \phi (\bm{x})\ge \phi (\bm{y})$.
\end{definition}

 \underline{\textbf{Optimal} $\bm{Q}_{\bm{X}}$}
 Based on the basic matrix inequalities and majorization theory together with the
 following EVDs (\ref{eq23}) to (\ref{eq26}) and the singular value decomposition (SVD)
 (\ref{eq27}) \setcounter{equation}{22}
\begin{align}
 \bm{F}^{\rm H}\bm{\Pi}\bm{F} =& \bm{U}_{\bm{F}\bm{\Pi}\bm{F}}\bm{\Lambda}_{\bm{F}\bm{\Pi}\bm{F}}
  \bm{U}_{\bm{F}\bm{\Pi}\bm{F}}^{\rm H} \ \text{with} \ \bm{\Lambda}_{\bm{F}\bm{\Pi}\bm{F}} \searrow ,
  \label{eq23} \\
 \bm{\Phi} =& \bm{U}_{\bm{\Phi}}\bm{\Lambda}_{\bm{\Phi}}\bm{U}_{\bm{\Phi}}^{\rm H} \ \text{with} \
  \bm{\Lambda}_{\bm{\Phi}} \searrow , \label{eq24} \\
 \bm{\Phi} =& \bar{\bm{U}}_{\bm{\Phi}}\bar{\bm{\Lambda}}_{\bm{\Phi}}\bar{\bm{U}}_{\bm{\Phi}}^{\rm H}
  \ \text{with} \ \bar{\bm{\Lambda}}_{\bm{\Phi}} \nearrow , \label{eq25} \\
 \bm{A}\bm{\Phi}^{-1}\bm{A}^{\rm H} =& \bm{U}_{\bm{A}\bm{\Phi}\bm{A}}\bm{\Lambda}_{\bm{A}\bm{\Phi}\bm{A}}
  \bm{U}_{\bm{A}\bm{\Phi}\bm{A}}^{\rm H} \ \text{with} \ \bm{\Lambda}_{\bm{A}\bm{\Phi}\bm{A}} \searrow ,
  \label{eq26} \\
 \bm{A} =& \bm{U}_{\bm{A}}\bm{\Lambda}_{\bm{A}}\bm{V}_{\bm{A}}^{\rm H} \ \text{with} \
  \bm{\Lambda}_{\bm{A}} \searrow , \label{eq27}
\end{align}
 the optimal unitary matrices $\bm{Q}_{\bm{X}}$ corresponding to the various objective
 functions can be derived and they are listed in the right column of Table~\ref{tab1}\footnote{Note that the solutions of \textbf{Obj.\,8} and \textbf{Obj.\,9} are derived based on high SNR approximation as the effects of $\alpha\bm{I}$ and $\bm{\Phi}$ are neglected.}. The detailed proofs are given in Appendix~\ref{appendix_table_II}.
 In the SVD (\ref{eq27}), $\bm{\Lambda}_{\bm{A}}$ contains the singular values of
 $\bm{A}$, while $\bm{U}_{\bm{A}}$ and $\bm{V}_{\bm{A}}$ are the corresponding left and
 right unitary matrices, respectively.

 In Table~\ref{tab1}, the unitary $\bm{U}_{\text{DFT}}$ for \textbf{Obj.\,5.1} is a
 discrete Fourier transform (DFT) matrix, and $\bm{U}_{\text{GMD}}$ for \textbf{Obj.\,6.1}
 is the unitary matrix that makes the diagonal elements of $\bm{L}$ identical, that is,
 $\bm{U}_{\text{GMD}}$ is the right unitary matrix of the geometric mean decomposition
 (GMD) of $\big(\bm{Q}_{\bm{X}}^{\rm H}\bm{F}^{\rm H}\bm{\Pi}\bm{F}\bm{Q}_{\bm{X}}+\alpha
 \bm{I}\big)^{-0.5}$. It is also worth highlighting that for \textbf{Obj.\,8} and
 \textbf{Obj.\,9}, in general, the closed-form optimal $\bm{Q}_{\bm{X}}$ cannot be derived,
 and only the approximated optimal solutions can be obtained at high signal-to-noise ratio
 (SNR) conditions.

\subsubsection{Optimization of $\bm{F}$}

 For \textbf{Opt.\,1.1}, given the optimal $\bm{Q}_{\bm{X}}$ in Table~\ref{tab1}, the objective functions in Table~\ref{tab1} are monotonically decreasing functions with respect to the eigenvalues of $\bm{F}^{\rm H}\bm{\Pi}\bm{F}$. Therefore, the optimal solutions of
 $\bm{F}$ fall in the Pareto optimal solution set of the following multi-objective
 optimization problem \cite{XingTSP201501}
\begin{align}\label{eq28}
 \textbf{Opt.\,1.2:} & \max\limits_{\bm{F}} \bm{\lambda}\big(\bm{F}^{\rm H}\bm{\Pi}\bm{F}\big) , {\rm s.t.} \psi_j(\bm{F})\le 0 , ~ 1\le j\le I,
\end{align}
 where $\bm{\lambda}(\bm{F}^{\rm H}\bm{\Pi}\bm{F})=[\lambda_1(\bm{F}^{\rm H}
 \bm{\Pi}\bm{F}\big)\cdots \lambda_N\big(\bm{F}^{\rm H}\bm{\Pi}\bm{F})]^{\rm T}$.
 Clearly, the optimal structure of $\bm{F}$ depends on both the objective function and  on the
 constraints. As discussed in \cite{XingTSP201501}, deriving the optimal structure of
 $\bm{F}$ for \textbf{Opt.\,1.2} corresponds to deriving the optimal structures of
 $\bm{F}$ for \textbf{Opt.\,1.1} for various objectives functions, including \textbf{Obj.\,1}
 to \textbf{Obj.\,15}.

 Since $\psi_j(\bm{F})$ is right unitarily-invariant, \textbf{Opt.\,1.2} is equivalent to the
 following matrix-monotonic optimization problem
\begin{align}\label{eq29}
 \textbf{Opt.\,1.3:} & \max\limits_{\bm{F}} & \bm{F}^{\rm H}\bm{\Pi}\bm{F}, {\rm s.t.} \psi_j(\bm{F})\le 0 , ~ 1\le j\le I.
\end{align}
 Generally, matrix-monotonic optimization maximizes a positive semi-definite matrix
 under certain power constraints. The fundamental idea of matrix-monotonic optimization is to extend the objective functions and solution sets to get in return more freedoms that can be exploited to simplify the analysis.
 The optimal solutions of \textbf{Opt.\,1.1} for
 the objective functions \textbf{Obj.\,1} to \textbf{Obj.\,15} are all in the Pareto
 optimal solution set of \textbf{Opt.\,1.3}. Since matrix-monotonic optimization
 derives the common structure of the Pareto optimal solution set of \textbf{Opt.\,1.3},
 the common optimal structures derived are exactly the structures of the optimal
 solutions of \textbf{Opt.\,1.1}. By taking advantage of these optimal structures,
 \textbf{Opt.\,1.1} can be substantially simplified.

 Interestingly, $\bm{F}^{\rm H}\bm{\Pi}\bm{F}$ can be interpreted as a matrix version SNR \cite{XingTSP201501}. Thus, based on \textbf{Opt.\,1.3} it can be concluded that various MIMO transceiver
 optimization problems maximize this matrix version SNR. When there are multiple data streams,
 maximizing the matrix version SNR inherently constitute a multi-objective optimization problem. In
 addition, each unitary matrix $\bm{Q}_{\bm{X}}$ corresponds to a specific implementation
 scheme. The focus of matrix-monotonic optimization is how to maximize the positive
 semi-definite matrix $\bm{F}^{\rm H}\bm{\Pi}\bm{F}$ under certain constraints. Different
 objective functions realize different tradeoffs among the multiple data streams, and
 matrix-monotonic optimization is a powerful tool that unifies the different constrained
 optimization problems with various objective functions.
 Specifically, based on matrix-monotonic optimization, the common properties of these
 objective functions are revealed, which are reflected on the optimal diagonalizable
 structures.

 These structures can transform complex optimization problems relying on matrix
 variables into much simpler ones with only vector variables. Thus case-by-case
 investigations for different objective functions are avoided. Since the optimal
 structure of $\bm{F}$ also depends on the specific form of the constraints, in the
 following, three right unitary invariant constraints are investigated, namely,
 shaping constraint \cite{XingTSP201502}, joint power constraint \cite{XingTSP201502} and multiple weighted power constraints \cite{XingTSP201601}.

 \underline{\textbf{Shaping Constraint}}
 For the shaping constraint, i.e., \textbf{Constraint 4}, \textbf{Opt.\,1.3} becomes
 the following optimization problem \cite{XingTSP201502}
\begin{align}\label{eq30}
 \textbf{Opt.\,1.4:} & {\max}_{\bm{F}} \bm{F}^{\rm H}\bm{\Pi}\bm{F}, {\rm s.t.} \bm{F}\bm{F}^{\rm H}\preceq \bm{R}_{\rm s}.
\end{align}
 The following lemma reveals the optimal structure of $\bm{F}$ for \textbf{Opt.\,1.4}
 with the shaping constraint.

\begin{lemma}\label{L1}
 When $\bm{R}_{\rm s}$ is attainable, i.e., the rank of $\bm{R}_{\rm s}$ is not higher than the number of columns and the number of rows in $\bm{F}$, the optimal solution $\bm{F}_{\rm opt}$ of \textbf{Opt.\,1.4} is a
 square root of $\bm{R}_{\rm s}$, i.e., $\bm{F}_{\rm opt}\bm{F}_{\rm opt}^{\rm H}=
 \bm{R}_{\rm s}$.
\end{lemma}

\begin{proof}
 Since the shaping constraint in \textbf{Opt.\,1.4} is right unitarily-invariant for $\bm{F}$, the objective is equivalent
 to maximizing $\bm{\lambda}\big(\bm{F}^{\rm H}\bm{\Pi}\bm{F}\big)$, which is in turn
 equivalent to maximizing $\bm{\lambda}\big(\bm{\Pi}^{1/2}\bm{F}\bm{F}^{\rm H}
 \bm{\Pi}^{1/2}\big)$. As $\bm{F}\bm{F}^{\rm H}\preceq \bm{R}_{\rm s}$, it can be concluded that
 $\bm{\lambda}\big(\bm{\Pi}^{1/2}\bm{F}\bm{F}^{\rm H}
 \bm{\Pi}^{1/2}\big) \preceq \bm{\lambda}\big(\bm{\Pi}^{1/2}\bm{R}_{\rm s}
 \bm{\Pi}^{1/2}\big)$, in which the equality holds when $\bm{F}\bm{F}^{\rm H}=\bm{R}_{\rm s}$.
 When the rank of $\bm{R}_{\rm s}$ is not higher
 than the number of columns and the number of rows in $\bm{F}$, the optimal solution
 $\bm{F}_{\rm opt}$ is a square root of $\bm{R}_{\rm s}$. It is worth noting that the square roots of $\bm{R}_{\rm s}$ are not unique. There are many square roots of $\bm{R}_{\rm s}$, however the different square roots have the same performance. We can choose an arbitrary square root of $\bm{R}_{\rm s}$ without performance loss.
\end{proof}

 \underline{\textbf{Joint Power Constraint}}
 Under the joint power constraint, \textbf{Constraint 6}, \textbf{Opt.\,1.3} can be
 rewritten as
\begin{align}\label{eq31}
 \textbf{Opt.\,1.5:} & {\max}_{\bm{F}}  \bm{F}^{\rm H}\bm{\Pi}\bm{F}, {\rm s.t.}  {\rm Tr}\big(\bm{F}\bm{F}^{\rm H}\big) \hspace{-1mm}\le\hspace{-1mm} P , \bm{F}\bm{F}^{\rm H}\preceq \tau \bm{I} .
\end{align}
 The optimal solution $\bm{F}_{\rm opt}$ for \textbf{Opt.\,1.5} is given in Lemma~\ref{L2}.

\begin{lemma}\label{L2}
 For \textbf{Opt.\,1.5} with the joint power constraint, the Pareto optimal solutions
 satisfy the following structure
\begin{align}\label{eq32}
 \bm{F}_{\rm opt} =& \bm{U}_{\bm{\Pi}}\bm{\Lambda}_{\bm{F}}\bm{U}_{\text{Arb}}^{\rm H} ,
\end{align}
 where the unitary matrix $\bm{U}_{\bm{\Pi}}$ is specified by the EVD
\begin{align}\label{eq33}
 \bm{\Pi} =& \bm{U}_{\bm{\Pi}}\bm{\Lambda}_{\bm{\Pi}}\bm{U}_{\bm{\Pi}}^{\rm H} \
  \text{with} \ \bm{\Lambda}_{\bm{\Pi}} \searrow ,
\end{align}
 every diagonal element of the rectangular diagonal matrix $\bm{\Lambda}_{\bm{F}}$ is
 smaller than $\sqrt{\tau}$, and $\bm{U}_{\text{Arb}}$ is an arbitrary unitary matrix
 having the appropriate dimension.
\end{lemma}

\begin{proof}
The proof is given in Appendix~\ref{appendix_1}.
\end{proof}

\begin{remark}\label{R1}
 For the optimization problem only under the sum power constraint, the optimal structure
 for $\bm{F}_{\rm opt}$ is also specified by (\ref{eq32}), where the sum of the diagonal
 elements of $\bm{\Lambda}_{\bm{F}}$ is no larger than $P$.
\end{remark}

 \underline{\textbf{Multiple Weighted Power Constraints}}
 Under the multiple weighted power constraints, \textbf{Opt.\,1.3} becomes
\begin{align}\label{eq34}
 \textbf{Opt.\,1.6:} \max\limits_{\bm{F}}  \bm{F}^{\rm H}\bm{\Pi}\bm{F},{\rm s.t.}{\rm Tr}\big(\bm{\Omega}_i\bm{F}\bm{F}^{\rm H}\big)\hspace{-1mm} \le \hspace{-1mm}P_i , 1\hspace{-1mm}\le\hspace{-1mm} i\hspace{-1mm}\le \hspace{-1mm}I .
\end{align}
 Note that the weighted power constraints are convex \cite{XingIET} and the detailed proof is given in Appendix \ref{appendix_convexity}. The weighted power constraints include both the sum power constraint
 and per-antenna power constraints as its special cases.  The optimal solution $\bm{F}_{\rm {opt}}$ for \textbf{Opt.\,1.6} is given in Lemma~\ref{L3}.

\begin{lemma}\label{L3}
 The Pareto optimal solutions of \textbf{Opt.\,1.6} satisfy the following structure
\begin{align}\label{eq36}
 \bm{F}_{\rm opt} =& \bm{\Omega}^{-\frac{1}{2}}\bm{U}_{\widetilde{\bm{\Pi}}}
  \bm{\Lambda}_{\widetilde{\bm{F}}}\bm{U}_{\rm Arb}^{\rm H} ,
\end{align}
 where $\bm{U}_{\rm Arb}$ is an arbitrary unitary matrix of appropriate dimension,
 $\bm{\Omega}=\sum_{i=1}^I\alpha_i\bm{\Omega}_i$, the nonnegative scalars
 $\alpha_i$ are the weighting factors that ensure that the constraints ${\rm Tr}\big(\bm{\Omega}_i
 \bm{F}\bm{F}^{\rm H}\big)\le P_i$ hold and they can be computed by classic sub-gradient methods,
 while the unitary matrix $\bm{U}_{\widetilde{\bm{\Pi}}}$ is specified by the EVD
\begin{align}\label{eq37}
 \bm{\Omega}^{-\frac{1}{2}}\bm{\Pi}\bm{\Omega}^{-\frac{1}{2}} =& \bm{U}_{\widetilde{\bm{\Pi}}}
  \bm{\Lambda}_{\widetilde{\bm{\Pi}}}\bm{U}_{\widetilde{\bm{\Pi}}}^{\rm H} \ \text{with} \
  \bm{\Lambda}_{\widetilde{\bm{\Pi}}} \searrow .
\end{align}
\end{lemma}

\begin{proof}
 See Appendix~\ref{Apa}.
\end{proof}
\begin{algorithm}[t]
	{\caption{The sub-gradient algorithm for finding the weighting factors $\alpha_i, \forall i$ }
		\label{algorithm_1}
		\begin{algorithmic}[1]
			\STATE Initialization: set iteration index as $I_{\rm{ite}}\!=\!0$;\\
			~~~~~~~~ set the maximum iteration number $I_{\max}$;
			\\ ~~~~~~~~~randomly set initial weighting parameters $\alpha_i^{(0)}$, \\ ~~~~~~~~~$\forall~i=1,\cdots,I$;
			\REPEAT
			\STATE  Solve the problem \eqref{eq34} to obtain ${\bm{F}}^{(I_{\rm{ite}})}$ given $\alpha_i^{(I_{\rm{ite}})}$;
			\STATE  Define the step size $t_{I_{\rm{ite}}}=\frac{c}{a+I_{\rm{ite}}\cdot b},~\{a,b,c\}>0$;
			\STATE  Update $\alpha_i^{(I_{\rm{ite}}\!+\!1)} \!\!=\![\alpha_i^{(I_{\rm{ite}})}\!+t_{I_{\rm{ite}}}(\text{Tr}(
			\bm{\Omega}_i{\bm{F}}^{(I_{\rm{ite}})}(\!{\bm{F}}^{(I_{\rm{ite}})})^H)$ $-\!P_i)]^{+},~ \forall i$;
			\STATE  Update $I_{\rm{ite}}=I_{\rm{ite}}+1$
			\UNTIL $\alpha_i^{(I_{\rm{ite}}-1)}\left(\text{Tr}(
					\bm{\Omega}_i
					{\bm{F}}^{(I_{\rm{ite}}\!-\!1)}({\bm{F
					}}^{(I_{\rm{ite}}\!-\!1)})^H)
					\!\!-\!\!P_i\right)\!\le \! \varepsilon_i, \forall i$ or $I_{ite} \le I_{\max}$, where $\varepsilon_i>0, \forall i $ is  sufficiently  small.
			\RETURN $\alpha_i,{\bm{F}}^{I_{ite}},\forall~i=1,\cdots,I.$
	\end{algorithmic}}
\end{algorithm}

\underline{\textbf{Specific Applications}} Three specific
  applications are given for each lemma. In wireline communications
  relying on the ubiquitous digital subscriber lines (DSL), the
  shaping constraint, i.e., spectral mask constraint, is the most
  important constraint used for limiting the crosstalk by forcing the
  users/services to have zero power outside their predefined spectral
  ranges~\cite{Palomar2004}.  In order to impose a maximum transmit
  power limit in the different transmit directions, the joint power
  constraint can be used~\cite{Dai2012}.  For per-antenna power
  constraints, the most representative application example is the
  beamforming design of C-RAN, where the signals are transmitted from
  distributed antennas~\cite{XingTSP201601}.

\subsection{Advantages of Matrix-Monotonic Optimization}\label{S2.5}

Matrix-monotonic optimization theory can simplify
 the optimization problem relying on matrix variables into a much simpler one manipulating only vector
 variables. Using matrix-monotonic optimization, for example, the optimal structure of
 the matrix variable $\bm{F}$ can be derived and the remaining optimization problem
 becomes a much simpler one that  optimizes the diagonal matrix $\bm{\Lambda}_{\bm{F}}$.
 For the various objective functions and constraints discussed previously, the optimal
 solutions of the diagonal elements of the diagonal matrix $\bm{\Lambda}_{\bm{F}}$ are
 in fact diverse variants of classic water-filling solutions \cite{General_Waterfilling}, which can be readily obtained
 straightforwardly based on the corresponding Karush-Kuhn-Tucker (KKT) conditions \cite[P244]{Boyd04}.

 In the existing literature, MIMO transceiver optimization problems are unified in the
 framework based on majorization theory  \cite{Majorization}. Our work is different
 from this existing framework in two perspectives. Firstly, in \cite{Majorization},
 linear and nonlinear transceiver optimization is
 considered separately. In our work, they are considered in the same framework.
 Additionally, in our work, more objective functions are considered. More importantly,
 the shaping constraint, joint power constraint and multiple weighted power
 constraints are considered in our work instead of merely the sum power constraint.

 For the multiple weighted power constraints, to the best of our knowledge, all the
 existing works are based on the KKT conditions. There are several limitations for
 these existing works. Firstly, this method is only applicable to mutual information maximization
 and MSE minimization. It cannot be used for more general objective functions. The
 method is not applicable for example to more complex systems, such as multi-hop
 AF MIMO relaying systems. Moreover, the KKT condition based methods also suffer
 from a serious weaknesses due to the fact that the KKT conditions are only
 necessary conditions for the optimal solutions. As discussed in \cite{XingTVT2016},
 the so-called turning-off effect and ambiguity effect usually perturb the KKT conditions based
 methods when deriving the optimal solutions. To overcome this problem, a widely used
 method is to consider the covariance matrix as a new variable in order to exploit its
 hidden convex nature. Unfortunately, the cost of adopting this approach is that
 the rank constraint has to be relaxed first. By contrast, our matrix-monotonic
 optimization framework does not suffer from these problems and has much wider
 applications.

\section{Bayes Robust Matrix-Monotonic Optimization}\label{S3}

 In wireless communication systems, the channel parameters have to be estimated. However, due to the uncertainty
 introduced both by noise and the time-varying nature of wireless channels, channel estimation
 errors inevitably exist \cite{Zhang2008}, the true channel matrix $\bm{H}$ can
 be expressed by the following Kronecker formula \cite{Pastore2012,Ding09}
\begin{align}\label{Channel_Error} 
 \bm{H} =& \widehat{\bm{H}} + \bm{H}_{\rm W}\bm{\Psi}^{\frac{1}{2}} ,
\end{align}
 where $\widehat{\bm{H}}$ is the estimated channel matrix and $\bm{H}_{\rm W}
 \bm{\Psi}^{\frac{1}{2}}$ is the channel estimation error, in which the elements of
 $\bm{H}_{\rm W}$ obey the independent and identical complex Gaussian distribution
 $\mathcal{CN}(0,1)$ and the covariance matrix $\bm{\Psi}$ of the channel estimate
 is a function of both the training sequence and of the channel estimator \cite{Ding09,Pastore2012}. It is worth noting that in
this section we focus our attention on the robust transceiver
design for the scenario, where both the source and destination have
imperfect CSI.
 Based on (\ref{Channel_Error}), for Bayes robust transceiver optimization, the matrix
 $\bm{\Pi}$ in the matrix-monotonic optimization can be expressed as \cite{XingTSP201501}
\begin{align}\label{eq43}
 \bm{\Pi} =& \widehat{\bm{H}}^{\rm H}\big(\sigma_n^2\bm{I} + {\rm Tr}\big(\bm{X}
  \bm{X}^{\rm H}\bm{\Psi}\big)\bm{I}\big)^{-1}\widehat{\bm{H}},
\end{align}
 where $\sigma_n^2$ is the additive white noise power in the data transmission.

 As a result, the generic Bayes robust matrix-variable optimization can be
 formulated as \cite{XingTSP201501}
\begin{align}\label{eq44}
\begin{array}{lcl}
 \textbf{Opt.\,2.1:} & \min\limits_{\bm{X}} & f\big(\bm{X}^{\rm H}\widehat{\bm{H}}^{\rm H}
  \bm{K}_{\rm n}^{-1}\widehat{\bm{H}}\bm{X}\big) , \\
 & {\rm{s.t.}} & \bm{K}_{\rm n}= \sigma_{\rm{n}}^2\bm{I} + {\rm Tr}\big(\bm{X}\bm{X}^{\rm H}
  \bm{\Psi}\big)\bm{I} , \\
 & & \psi_i(\bm{X})\le 0 , ~ 1\le i \le I .
\end{array}
\end{align}
 As discussed in \cite{XingTSP201501}, after introducing the transformation $\bm{X}=\bm{F}
 \bm{Q}_{\bm{X}}$ and recalling that the constraints $\psi_i(\cdot )$ are right unitarily-invariant, \textbf{Opt.\,2.1} is transferred equivalently to the following matrix-monotonic
 optimization problem:
\begin{align}\label{Matrix_Monotonic_Per} 
\begin{array}{lcl}
 \textbf{Opt.\,2.2:} & \max\limits_{\bm{F}} & \bm{F}^{\rm H}\widehat{\bm{H}}^{\rm H}
  \bm{K}_{\rm n}^{-1}\widehat{\bm{H}}\bm{F} , \\
 & {\rm{s.t.}} & \bm{K}_{\rm n}=\sigma_{\rm{n}}^2\bm{I}+{\rm Tr}\big(\bm{F}\bm{F}^{\rm H}
  \bm{\Psi}\big)\bm{I} ,  \\
 & & \psi_i(\bm{F})\le 0 , 1\le i\le I .
\end{array}
\end{align}
 Here the matrix $\bm{F}^{\rm H}\widehat{\bm{H}}^{\rm H}\bm{K}_{\rm n}^{-1}
 \widehat{\bm{H}}\bm{F}$ can be regarded as an extended SNR matrix in the presence of channel estimation errors, and this kind of matrix-monotonic optimization is named as robust matrix-monotonic optimization in \cite{XingTSP201501}. In the following, we discuss the
 optimal solutions of this robust matrix-monotonic optimization problem under specific
 power constraints.

{\emph{1)}~{\rm\bf Shaping Constraint}}:
 Consider the shaping constraint of
\begin{align}\label{eq46}
 \psi_1 (\bm{F}) =& \bm{F}\bm{F}^{\rm H} - \bm{R}_{\rm s}.
\end{align} As proved in Appendix~\ref{appendix_3}, for the general case of $\bm{\Psi} \not\propto \bm{I}$, a suboptimal solution for \textbf{Opt.\,2.2} which maximizes a lower bound of the objective of \textbf{Opt.\,2.2} is given by Lemma~\ref{L1}.
When $\bm{\Psi}=\bm{0}$, the lower bound is tight and the solution given in Lemma~\ref{L1} is exactly the Pareto optimal solution of \textbf{Opt.\,2.2}.


{\emph{2)}~{\rm\bf Joint Power Constraint}}:
 Next consider the joint power constraint specified by
\begin{align}\label{eq47}
 \psi_1(\bm{F}) = {\rm Tr}\big(\bm{F}\bm{F}^{\rm H}\big) -P , \
 \psi_2(\bm{F}) = \bm{F}\bm{F}^{\rm H} - \tau \bm{I} .
\end{align}
 For the perfect CSI case associated with $\bm{\Psi}=\bm{0}$, the Pareto optimal solutions of
 \textbf{Opt.\,2.2} are specified by Lemma~\ref{L2}. When $\bm{\Psi}\propto \bm{I}$
 and $\psi_1(\bm{F})\le 0$ is active at the optimal solutions $\bm{F}_{\rm opt}$,
 the Pareto optimal solutions of \textbf{Opt.\,2.2} also satisfy the structure given
 in Lemma~\ref{L2}, since in this case $\bm{K}_{\rm n}$ is constant. As proved in Appendix~\ref{appendix_3}, for the general case $\bm{\Psi}\not\propto \bm{I}$, the suboptimal solution that maximizes a lower bound of the objective of \textbf{Opt.\,2.2} satisfies the following structure
\begin{align}
\hspace*{-2mm}\bm{F}\hspace*{-1mm}=\hspace*{-1mm}\frac{\sigma_{\rm{n}}
\widetilde{\bm{\Psi}}^{-\frac{1}{2}}
{\bm{V}}_{\widetilde{\bm{H}}}\bm{\Lambda}_{\widetilde{\bm{F}}}{\bm{U}}_{\rm{Arb}}^{\rm{H}}}
{\left(1\hspace*{-1mm}-\hspace*{-1mm}
{\rm{Tr}}\left(\widetilde{\bm{\Psi}}^{-\frac{1}{2}}
\bm{\Psi}\widetilde{\bm{\Psi}}^{-\frac{1}{2}}{\bm{V}}_{\widetilde{\bm{H}}}
\bm{\Lambda}_{\widetilde{\bm{F}}}\bm{\Lambda}_{\widetilde{\bm{F}}}^{\rm{H}}
{\bm{V}}_{\widetilde{\bm{H}}}^{\rm{H}}\right)\right)^{\frac{1}{2}}},
\end{align} where $\widetilde{\bm{\Psi}}=\sigma_{\rm{n}}^2\bm{I}+P\bm{\Psi}$ and the unitary matrix ${\bm{V}}_{\widehat{\bm{H}}}$ is defined based on the following SVD \begin{align}
\widehat{\bm{H}}\big(\sigma_{\rm{n}}^2\bm{I}+P\bm{\Psi}\big)^{-\frac{1}{2}}=
   {\bm{U}}_{\widetilde{\bm{H}}}\bm{\Lambda}_{\widetilde{\bm{H}}}{\bm{V}}_{\widetilde{\bm{H}}}^H, \text{with}, \bm{\Lambda}_{\widetilde{\bm{H}}} \searrow.
\end{align} The diagonal elements of the rectangular diagonal matrix $\bm{\Lambda}_{\widetilde{\bm{F}}_k}$ are smaller than
$\sqrt{\tau{(\sigma_{\rm{n}}^2+P\lambda_{\min}(\bm{\Psi}))}/{(\sigma_{\rm{n}}^2+P\lambda_{\max}(\bm{\Psi}))}}$.

{\emph{3)}~{\rm\bf Multiple Weighted Power Constraints}}:
 When the multiple weighted power constraints are used, we have
\begin{align}\label{eq48}
\psi_i (\bm{F})= {\rm Tr}\big(\bm{\Omega}_i\bm{F}\bm{F}^{\rm H}\big) - P_i , ~ 1\le i\le I .
\end{align}
 From ${\rm Tr}\big(\bm{\Omega}_i\bm{F}\bm{F}^{\rm H}\big)\le P_i$, it is readily seen that the following inequality holds
\begin{align}\label{eq49}
 {\rm Tr}\big(\big(\sigma_{\rm{n}}^2\bm{\Omega}_i+P_i\bm{\Psi}\big)\bm{F}\bm{F}^{\rm H}\big) =&
  {\rm Tr}\big(\sigma_n^2\bm{\Omega}_i\bm{F}\bm{F}^{\rm H}\big) + P_i{\rm Tr}\big(\bm{\Psi}
  \bm{F}\bm{F}^{\rm H}\big) \nonumber \\
 \le & \sigma_n^2 P_i + P_i {\rm Tr}\big(\bm{\Psi}\bm{F}\bm{F}^{\rm H}\big) .
\end{align}
 Hence ${\rm Tr}\big(\bm{\Omega}_i\bm{F}\bm{F}^{\rm H}\big)\le P_i$ is equivalent to
\begin{align}\label{eq50}
 \frac{{\rm Tr}\big[\big(\sigma_{\rm{n}}^2\bm{\Omega}_i+P_i\bm{\Psi}\big)\bm{F}\bm{F}^{\rm H}\big]}
  {\sigma_{\rm{n}}^2+{\rm Tr}\big(\bm{F}\bm{F}^{\rm H}\bm{\Psi}\big)}\le P_i.
\end{align}
 As a result, the Bayes robust matrix-monotonic optimization problem
 (\ref{Matrix_Monotonic_Per}) is equivalent to the following problem
\begin{align}
\label{Matrix_Monotonic_Per_A} 
\hspace*{-2mm}\begin{array}{lcl}
 \textbf{Opt.\,2.3:} \!\! &\!\! \max\limits_{\bm{F}}\!\! &\!\! \bm{F}^{\rm H}\widehat{\bm{H}}^{\rm H}
  \bm{K}_{\rm n}^{-1}\widehat{\bm{H}}\bm{F} , \\
 \!\! &\!\! {\rm{s.t.}}\!\! &\!\! \bm{K}_{\rm n}=\sigma_n^2\bm{I}+{\rm Tr}\big(\bm{F}\bm{F}^{\rm H}
  \bm{\Psi}\big)\bm{I} ,  \\
 \!\! &\!\! \!\! &\!\!  \frac{{\rm Tr}\big(\big(\sigma_{\rm{n}}^2\bm{\Omega}_i+P_i\bm{\Psi}\big)\bm{F}\bm{F}^{\rm H}\big)}
  {\sigma_{\rm{n}}^2+{\rm Tr}\big(\bm{F}\bm{F}^{\rm H}\bm{\Psi}\big)}\le P_i, 1\le i\le I .
\end{array}\!\!
\end{align}
 By defining the auxiliary matrix variable
\begin{align}\label{F_hat} 
 \bar{\bm{F}} =& {\big[\sigma_{\rm{n}}^2+{\rm Tr}\big(\bm{F}\bm{F}^{\rm H}\bm{\Psi}\big)\big]^{-\frac{1}{2}}}
  \bm{F} ,
\end{align}
 the optimization problem (\ref{Matrix_Monotonic_Per_A}) can be simplified to:
\begin{align}\label{eq53}
\hspace*{-3mm}\begin{array}{lcl}
 \textbf{Opt.\,2.4:}\!\! &\!\! \max\limits_{\bar{\bm{F}}}\!\! &\!\!\! \bar{\bm{F}}^{\rm H}\widehat{\bm{H}}^{\rm H}
  \widehat{\bm{H}}\bar{\bm{F}} , \\
 \!\! & \!\!{\rm{s.t.}}\!\! &\!\!\! {\rm Tr}\big(\! \big(\sigma_{\rm{n}}^2\bm{\Omega}_i\! +\! P_i\bm{\Psi}\big)\bar{\bm{F}}
  \bar{\bm{F}}^{\rm H}\big)\! \le\! P_i, \, 1\! \le\! i\! \le\! I .
\end{array}\!\!
\end{align}
 Similar to the proof of Lemma~\ref{L3}, specifically to (\ref{label_opt}) in Appendix~\ref{Apa},
 the above optimization problem is equivalent to
\begin{align}\label{eq54}
 \hspace*{-3mm}\textbf{Opt.\,2.5:} & {\max}_{\bar{\bm{F}}} \bar{\bm{F}}^{\rm H}\widehat{\bm{H}}^{\rm H}
  \widehat{\bm{H}}\bar{\bm{F}}, {\rm{s.t.}}  {\rm Tr}\big(\bar{\bm{\Omega}}\bar{\bm{F}}\bar{\bm{F}}^{\rm H}\big)\hspace*{-1mm}\le {\sum}_{i=1}^I \hspace*{-1mm}P_i,
\end{align}
 where
\begin{align}\label{eq55}
 \bar{\bm{\Omega}} =& \sum\nolimits_{i=1}^I \alpha_i\big(\sigma_{\rm{n}}^2\bm{\Omega}_i+P_i\bm{\Psi}\big) .
\end{align}

 According to Lemma~\ref{L3}, the Pareto optimal solutions $\bar{\bm{F}}_{\rm opt}$ of
 \textbf{Opt.\,2.5} satisfy the following structure
\begin{align}\label{eq56}
 \bar{\bm{F}}_{\rm opt} =& \bar{\bm{\Omega}}^{-\frac{1}{2}} \bm{V}_{\bm{\mathcal{H}}}
  \bm{\Lambda}_{\widetilde{\bar{\bm{F}}}} \bm{U}_{\text{Arb}}^{\rm H} ,
\end{align}
 where the unitary matrix $\bm{V}_{\bm{\mathcal{H}}}$ is specified by the SVD of:
\begin{align}\label{eq57}
 \widehat{\bm{H}}\bar{\bm{\Omega}}^{-\frac{1}{2}} =& \bm{U}_{\bm{\mathcal{H}}}
  \bm{\Lambda}_{\bm{\mathcal{H}}} \bm{V}_{\bm{\mathcal{H}}}^{\rm H}.
\end{align} From
 (\ref{F_hat}), we have $\big[\sigma_{\rm{n}}^2+{\rm Tr}\big(\bm{F}\bm{F}^{\rm H}\bm{\Psi}\big)
 \big]^{\frac{1}{2}}\bar{\bm{F}}=\bm{F}$ and based on this conclusion we have the following equation
\begin{align}\label{eq58}
 \big[\sigma_{\rm{n}}^2\! +\! {\rm Tr}\big(\bm{F}\bm{F}^{\rm H}\bm{\Psi}\big)\big]{\rm Tr}\big(\bm{\Psi}
  \bar{\bm{F}}\bar{\bm{F}}^{\rm H}\big)\! +\! \sigma_{\rm{n}}^2 = {\rm Tr}\big(\bm{\Psi}\bm{F}
  \bm{F}^{\rm H}\big)\! +\! \sigma_{\rm{n}}^2 .
\end{align}
 This yields
\begin{align}\label{eq59}
 \sigma_{\rm{n}}^2+{\rm Tr}\big(\bm{F}\bm{F}^{\rm H}\bm{\Psi}\big) =& {\sigma_{\rm{n}}^2}/\left({1 -
  {\rm Tr}\big(\bm{\Psi}\bar{\bm{F}}\bar{\bm{F}}^{\rm H}\big)}\right) .
\end{align}
 Thus, given the Pareto optimal $\bar{\bm{F}}_{\rm opt}$, the Pareto optimal
 $\bm{F}_{\rm opt}$ is expressed as
\begin{align}\label{eq60}
 \bm{F}_{\rm opt} =&\sqrt{{\sigma_{\rm{n}}^2}/[{1-{\rm Tr}\big(\bm{\Psi}\bar{\bm{F}}_{\rm opt}
  \bar{\bm{F}}_{\rm opt}^{\rm H}\big)}]} \, \bar{\bm{F}}_{\rm opt} .
\end{align}
 Given (\ref{eq60}) and (\ref{eq56}), we arrive at the following lemma.

\begin{lemma}\label{L4}
 The Pareto optimal solutions $\bm{F}_{\rm opt}$ of $\textbf{Opt.\,2.2}$ under the multiple
 weighted power constraints satisfy the following structure
\begin{align}\label{eq61}
 \bm{F}_{\rm opt} &= \frac{\sigma_{\rm{n}} \bar{\bm{\Omega}}^{-\frac{1}{2}} \bm{V}_{\bm{\mathcal{H}}}
  \bm{\Lambda}_{\widetilde{\bar{\bm{F}}}} \bm{U}_{\rm{Arb}}^{\rm H} }
  {\big[1- {\rm Tr}\big(\bar{\bm{\Omega}}^{-\frac{1}{2}} \bm{\Psi} \bar{\bm{\Omega}}^{-\frac{1}{2}}
  \bm{V}_{\bm{\mathcal{H}}} \bm{\Lambda}_{\widetilde{\bar{\bm{F}}}} \bm{\Lambda}_{\widetilde{\bar{\bm{F}}}}^{\rm H}
  \bm{V}_{\bm{\mathcal{H}}}^{\rm H}\big)\big]^{\frac{1}{2}} } .
\end{align}
\end{lemma}

\begin{figure*}[tp!]\setcounter{equation}{62}
\vspace*{-3mm}
\begin{align}\label{eq66}
\hspace*{-4mm} \textbf{Opt.\,3.3:} & \min\limits_{\bm{F},\bm{Q}_{\bm{X}}}  \int f\Big(\bm{Q}_{\bm{X}}^{\rm{H}} \bm{F}^{\rm{H}}( \widehat{\bm{H}} \hspace*{-1mm}+\hspace*{-1mm} \bm{\Sigma}^{\frac{1}{2}} \bm{H}_{\rm W}
  \bm{\Psi}^{\frac{1}{2}} )^{\rm H} \bm{R}_{\rm n}^{-1}  (\widehat{\bm{H}} \hspace*{-1mm}+\hspace*{-1mm}\bm{\Sigma}^{\frac{1}{2}}
  \bm{H}_{\rm W} \bm{\Psi}^{\frac{1}{2}}) \bm{F}\bm{Q}_{\bm{X}}  \Big) p\big( \bm{H}_{\rm W} \big)
  {\rm{d}} \bm{H}_{\rm W} ,  {\rm{s.t.}} \psi_i(\bm{F}\big)\le 0 , 1\le i\le I ,
\end{align}
\hrulefill
\begin{align}\label{eq67}
\hspace*{-4mm} \textbf{Opt.\,3.4:} & \min\limits_{\bm{F},\bm{Q}_{\bm{X}}}  {f}\Big( \bm{Q}_{\bm{X}}^{\rm{H}}  \int \bm{F}^{\rm{H}}\big(
  \widehat{\bm{H}}\hspace*{-1mm} +\hspace*{-1mm}\bm{\Sigma}^{\frac{1}{2}} \bm{H}_{\rm W}
  \bm{\Psi}^{\frac{1}{2}}\big)^{\rm H}
  \bm{R}_{\rm n}^{-1}  \big(\widehat{\bm{H}}\hspace*{-1mm} +\hspace*{-1mm}\bm{\Sigma}^{\frac{1}{2}} \bm{H}_{\rm W}
  \bm{\Psi}^{\frac{1}{2}} \big)\bm{F} p\big(\bm{H}_{\rm W}\big ){\rm{d}} \bm{H}_{\rm W}
\bm{Q}_{\bm{X}} \Big),
   {\rm{s.t.}}  \psi_i(\bm{F}\big)\le 0, 1\le i\le I.
\end{align}
\hrulefill
\begin{align}\label{eq68}
 \textbf{Opt.\,3.5:} & \min\limits_{\bm{F}}  \bm{\lambda} \Big(\mathbb{E}\{\bm{F}^{\rm H} (\bm{\Psi}^{\frac{1}{2}} \bm{H}_{\rm W}^{\rm H} \bm{\Sigma}^{\frac{1}{2}}+\widehat{\bm{H}}^{\rm H} )
 \bm{R}_{\rm n}^{-{1}} (\bm{\Sigma}^{\frac{1}{2}} \bm{H}_{\rm W} \bm{\Psi}^{\frac{1}{2}}+\widehat{\bm{H}} )
  \bm{F}\}  \Big),
  {\rm{s.t.}}  \psi_i(\bm{F}\big)\le 0 , 1\le i\le I .
\end{align}
\hrulefill
\vspace*{-4mm}
\end{figure*}

 The robust optimal structure under the multiple weighted power constraints given in
 Lemma~\ref{L4} is significantly different from the existing conclusions previously designed for the robust
 solutions under the sum power constraints \cite{XingTSP201501} and for the transceiver designs
  relying on perfect CSI under the per-antenna power constraints \cite{XingTSP201601}. In the
  traditional robust transceiver designs under the sum power
  constraint, there is no restriction on the unitary matrix in the
  eigenvalue decomposition of the covariance matrix
  $\bm{F}\bm{F}^{\rm{H}}$. For the multiple weighted power
  constraints, there is a restriction on the unitary matrix in the
  eigenvalue decomposition. Then a new rotation matrix $
  \bar{\bm{\Omega}}^{-\frac{1}{2}}$ is needed, based on which the
  precoder matrix can align the direction with the space constructed
  by the multiple weighting matrices.

\section{Stochastically Robust Matrix-Monotonic Optimization}\label{S4}

 When  CSI at the receiver (CSIR) is  perfect, but  only   statistical CSI at the  transmitter (CSIT) is available, the
 corresponding stochastically robust matrix-monotonic optimization can be formulated as
 \cite{Jafar2005} \setcounter{equation}{58}
\begin{align}\label{eq62}
\begin{array}{lcl}
 \textbf{Opt.\,3.1:} & \min\limits_{\bm{X}} & \mathbb{E}_{\bm{H}}\big\{ f\big(\bm{X}^{\rm H}
  \bm{H}^{\rm H}\bm{R}_{\rm n}^{-1}\bm{H}\bm{X}\big)\big\} , \\
 & {\rm{s.t.}} & \psi_i(\bm{X})\le 0 , ~ 1\le i\le I ,
\end{array}
\end{align}
 where $\bm{R}_{\rm n}$ is the noise covariance matrix.  For simplicity, we mainly consider $\bm{R}_{\rm n}=\sigma_n^2\bm{I}$ as that in Section~\ref{S3}.  For this kind of
 optimization problems, the objective function is an average value over the
 distribution of the channel matrix $\bm{H}$ modeled by
 \begin{align}\label{H_Stochastic} 
 \bm{H} =&\widehat{\bm{H}} + \bm{\Sigma}^{\frac{1}{2}}\bm{H}_{\rm W}\bm{\Psi}^{\frac{1}{2}} ,
 \end{align}
 where  $\bm{\Sigma}$ and $\bm{\Psi}$ are the row
 and column correlation matrices, respectively. For MIMO systems, $\bm{\Sigma}$ is the
 spatial correlation matrix of the receiver antenna array, while $\bm{\Psi}$ is the
 spatial correlation matrix of the transmitter antenna array. Since the constraints are
 right unitarily-invariant, \textbf{Opt.\,3.1} can be expressed as
\begin{align}\label{eq64}
\begin{array}{lcl}
 \textbf{Opt.\,3.2:}\!\! &\!\! \min\limits_{\bm{F}}\!\! &\!\! \mathbb{E}_{\bm{H}}\big\{\! f\big(\bm{Q}_{\bm{X}}^{\rm H}
  \bm{F}^{\rm H}\bm{H}^{\rm H}\bm{R}_{\rm n}^{-1}\bm{H}\bm{F}\bm{Q}_{\bm{X}}\big)\! \big\} , \\
 \!\! &\!\! {\rm{s.t.}}\!\! &\!\! \psi_i(\bm{F})\le 0 , ~ 1\le i\le I .
\end{array}
\end{align}

 The stochastically robust matrix-monotonic optimization naturally aims at optimizing the
 distribution of the random matrix $\bm{H}\bm{F}$, based on the channel model (\ref{H_Stochastic}).
 Therefore, \textbf{Opt.\,3.2} can be rewritten as \textbf{Opt.\,3.3} of (\ref{eq66})
 at the top of the next page, where $p\big(\bm{H}_{\rm W}\big)$ is the probability
 density function (PDF) of $\bm{H}_{\rm W}$. As pointed out in \cite{Zhang2008}, the analytical  expression  of the average value of  an  arbitrary objective function ${f}(\cdot)$ in  \textbf{Opt.\,3.3}  is  impossible to obtain, which  thus makes \textbf{Opt.\,3.3} difficult  to address.  An alternative scheme is to consider the average matrix  in the objective function ${f}(\cdot)$ and the corresponding optimization problem is \textbf{Opt.\,3.4} given in (\ref{eq67}).
 Taking
 \begin{align}
 \bm{\Pi}= \mathbb{E} \{\big(
  \widehat{\bm{H}} +\bm{\Sigma}^{\frac{1}{2}} \bm{H}_{\rm W}
  \bm{\Psi}^{\frac{1}{2}}\big)^{\rm H}
  \bm{R}_{\rm n}^{-1}  \big(\widehat{\bm{H}} +\bm{\Sigma}^{\frac{1}{2}} \bm{H}_{\rm W}
  \bm{\Psi}^{\frac{1}{2}} \big) \}
 \end{align} the optimal solutions of $\bm{Q}_{\bm{X}}$ are given in Table~\ref{tab1}. Based on the optimal solutions of $\bm{Q}_{\bm{X}}$. Similar to \textbf{Opt.\,1.2}, the optimal solutions $\bm{F}$ of \textbf{Opt.\,3.4} fall in the Pareto optimal solution set of \textbf{Opt.\,3.5} in (\ref{eq68}).
 It is obvious that
 the key to the optimization of  \textbf{Opt.\,3.5} is to maximize
 eigenvalues of $\mathbb{E}[ \bm{F}^{\rm H} (\bm{\Psi}^{\frac{1}{2}} \bm{H}_{\rm W}^{\rm H} \bm{\Sigma}^{\frac{1}{2}}+\widehat{\bm{H}}^{\rm H} )
 \bm{R}_{\rm n}^{-{1}} (\bm{\Sigma}^{\frac{1}{2}} \bm{H}_{\rm W} \bm{\Psi}^{\frac{1}{2}}+\widehat{\bm{H}} )
 \bm{F}]= \bm{F}^{\rm H}(\widehat{\bm{H}}^{\rm H} \bm{R}_{\rm n}^{-1} \widehat{\bm{H}} +\text{Tr}(\bm{\Sigma}\bm{R}_{\rm n}^{-{1}} )\bm{\Psi})\bm{F}$. Hence, this stochastically robust optimization problem \textbf{Opt.\,3.5} is equivalent to\setcounter{equation}{65}\begin{align}\label{eq69}
\begin{array}{lcl}
 \hspace*{-4mm} \textbf{Opt.\,3.6:}\hspace{-3mm}   & \hspace{-3mm} \max\limits_{\bm{F}} &\! \! \! \hspace{-1mm} \bm{\lambda}\big( \bm{F}^{\rm H}(\widehat{\bm{H}}^{\rm H} \bm{R}_{\rm n}^{-1} \widehat{\bm{H}} \! \! +\!\!  \text{Tr}(\bm{\Sigma}\bm{R}_{\rm n}^{-{1}} )\bm{\Psi})\bm{F}\big), \\
 &{\rm{s.t.}} & \hspace{-2mm} \psi_i(\bm{F})\le 0 ,  1\le i\le I.
\end{array}
\end{align}
 As discussed previously in Section~\ref{S2.3}, the above multi-objective optimization problem
 is equivalent to the following matrix-monotonic optimization problem
\begin{align}\label{eq70}
\begin{array}{lcl}
 \textbf{Opt.\,3.7:}\!\! &\! \max\limits_{\bm{F}} &  \!\bm{F}^{\rm H}(\widehat{\bm{H}}^{\rm H} \bm{R}_{\rm n}^{-1} \widehat{\bm{H}} \!+\!\text{Tr}(\bm{\Sigma}\bm{R}_{\rm n}^{-{1}} )\bm{\Psi})\bm{F}, \\
 & {\rm{s.t.}} & \!\psi_i(\bm{F})\le 0 ,  1\le i\le I.
\end{array}
\end{align}
 Again, we discuss the Pareto optimal solutions of this stochastically robust matrix-monotonic
 optimization problem under three specific power constraints, respectively.

{\emph{1)}~{\rm\bf Shaping Constraint}}:
 Under the shaping constraint of $\psi (\bm{F})=\bm{F}\bm{F}^{\rm H}-\bm{R}_{\rm s}$,
 the optimal solution $\bm{F}_{\rm opt}$ to \textbf{Opt.\,3.7} is specified
 by Lemma~\ref{L1}. Specifically, when the rank of $\bm{R}_{\rm s}$ is not higher than
 the number of columns and the number of rows in $\bm{F}$, $\bm{F}_{\rm opt}$ is a
 square root of $\bm{R}_{\rm s}$.

{\emph{2)}~{\rm\bf Joint Power Constraint}}:
 Clearly, under the joint power constraint (\ref{eq47}), \textbf{Opt.\,3.7} is identical
 to \textbf{Opt.\,1.5} with
 \begin{align}
 \label{Pi}
 {\bm{\Pi}}=\widehat{\bm{H}}^{\rm H} \bm{R}_{\rm n}^{-1} \widehat{\bm{H}} \!+\!\text{Tr}(\bm{\Sigma}\bm{R}_{\rm n}^{-{1}} )\bm{\Psi}.
 \end{align}
 Therefore, the Pareto optimal solutions
 $\bm{F}_{\rm opt}$ of \textbf{Opt.\,3.7} under the joint power constraint are defined
 exactly in Lemma~\ref{L2} by simply replacing $\bm{\Pi}$ in (\ref{eq32}) and (\ref{eq33})
 with (\ref{Pi}).

{\emph{3)}~{\rm\bf Multiple Weighted Power Constraints}}:
 Obviously, under the multiple weighted power constraints (\ref{eq48}), the Pareto optimal
 solutions $\bm{F}_{\rm opt}$ of \textbf{Opt.\,3.7} are specified by Lemma~\ref{L3}, where
 $\bm{\Pi}$ should be replaced by (\ref{Pi}).

\section{Worst Case Robust Matrix-Monotonic Optimization}\label{S5}
In this section, we  consider  the   norm bounded
CSI error, i.e., $\|\Delta\bm{H}\|_2\le \gamma$
 with $\| \cdot \|_2$ denoting  matrix Spectral norm, and
adopt   the worst case (min-max)
 criterion as a figure-of-merit for robust designs \cite{JHWang2013}.  Hereafter, spectral norm is adopted since it can act as  both  lower and upper bounds of   the widely adopted Frobenius and Nuclear norms and is generally tractable \cite{JHWang2013,ShiqiUAV}.
 For example, we have  $\| \cdot \|_2\! \le\! \| \cdot \|_F\! \le\!
 \sqrt{\text{rank}(\cdot)}\| \cdot \|_2$,  implying  that  spectral norm
 constrained  error can also provide valuable insights for Frobenius norm constrained case. Moreover, for the same error size, spectral norm  generally  covers the largest error region.

 Let us denote the estimated channel matrix and channel
 error matrix by ${\widehat{\bm{H}}}$ and $\Delta\bm{H}$.  With
 norm-bounded CSI error $\|\Delta\bm{H}\|_2\le \gamma$, substituting
 ${\bm{H}}={\widehat{\bm{H}}}-\Delta\bm{H}$ and
 $\bm{R}_n=\sigma_n^2{\bm{I}}$ into \textbf{Opt.\,1.3}, the robust matrix-monotonic optimization problem under Spectral
norm bounded CSI error can be
 formulated as
\begin{align}\label{eq71}
\begin{array}{lcl}
 \hspace{-5mm}\textbf{Opt.\,4.1:}\!\!\! &\!\! \max\limits_{\bm{F}}\min\limits_{\Delta\bm{H}}\!\! &\!\! \sigma_n^{-2}\bm{F}^{\rm H}
  \big(\widehat{\bm{H}}\! -\! \Delta\bm{H}\big)^{\rm H}
  \big(\widehat{\bm{H}}\! -\! \Delta\bm{H}\big) \bm{F} \! , \\
 \!\!\! &\!\! {\rm{s.t.}}\!\! &\!\! \psi_i(\bm{F})\le 0 , 1 \hspace{-1mm}\le \hspace{-1mm} i\le I,  \|\Delta\bm{H}\|_2\le  \gamma .
\end{array} \!
\end{align}

The Spectral norm is unitarily-invariant, which
  means that for the arbitrary $\Vert \bm{E}\Vert_2 \le \epsilon_s$,
  it yields $\Vert\bm{U} \bm{E}\bm{V}\Vert_2 \le \epsilon_s$ given any
  unitary matrices $\bm{U}$ and $\bm{V}$.  Based on the following
  matrix inequality \cite[P471]{Horn1990}
  \begin{align}
  \big(\sigma_i(\bm{B})-\sigma_1(\bm{C})\big)^{+} \le  \sigma_i(\bm{B}+\bm{C})
  \end{align} we readily conclude that
  \begin{align}
  \sigma_i(\widehat{\bm{H}}\! -\! \Delta\bm{H})\ge  \big(\sigma_i(\widehat{\bm{H}})-\sigma_1(\Delta\bm{H})\big)^{+}.
  \end{align} Therefore, we have the following eigenvalue inequality
   \begin{align}
   &{\bm{\lambda}}\left(\big(\widehat{\bm{H}}\! -\! \Delta\bm{H}\big)^{\rm H}
   \big(\widehat{\bm{H}}\! -\! \Delta\bm{H}\big)\right) \nonumber \\
   & \hspace{-2mm}\succeq {\bm{\lambda}}\left(\big(\widehat{\bm{H}}\! -\! {\bm{U}}_{\widehat{\bm{H}}}\bm{\Lambda}_{\Delta{\bm{H}}}{\bm{V}}_{\widehat{\bm{H}}}^H \big)^{\rm H}
   \big(\widehat{\bm{H}}\! -\!  {\bm{U}}_{\widehat{\bm{H}}}  \bm{\Lambda}_{\Delta{\bm{H}}} {\bm{V}}_{\widehat{\bm{H}}}^H\big) \right)\end{align} where ${\bm{U}}_{\widehat{\bm{H}}}$ and ${\bm{V}}_{\widehat{\bm{H}}}$ are derived from the following SVD
   \begin{align}
   \label{worst_H}
   \widehat{\bm{H}}={\bm{U}}_{\widehat{\bm{H}}}\bm{\Lambda}_{\widehat{\bm{H}}}
   {\bm{V}}_{\widehat{\bm{H}}}^H, \text{with},\bm{\Lambda}_{\widehat{\bm{H}}} \searrow.
    \end{align} The diagonal matrix $\bm{\Lambda}_{\Delta{\bm{H}}}$ equals
   \begin{align}
    [\bm{\Lambda}_{\Delta{\bm{H}}}]_{i,i}=\min\left([\bm{\Lambda}_{\widehat{\bm{H}}}]_{i,i}, {\gamma}\right), \forall i.
   \end{align} Then there exists a unitary matrix $\bm{Q}$ makes the following matrix inequality hold
   \begin{align}
   \label{matrix_inequality_singular}
   &\bm{Q}^{\rm{H}}\big(\widehat{\bm{H}}\! -\! \Delta\bm{H}\big)^{\rm H}
   \big(\widehat{\bm{H}}\! -\! \Delta\bm{H}\big)\bm{Q} \nonumber \\
   & \hspace{-2mm}\succeq \big(\widehat{\bm{H}}\! -\! {\bm{U}}_{\widehat{\bm{H}}}\bm{\Lambda}_{\Delta{\bm{H}}}{\bm{V}}_{\widehat{\bm{H}}}^H \big)^{\rm H}
   \big(\widehat{\bm{H}}\! -\!  {\bm{U}}_{\widehat{\bm{H}}}  \bm{\Lambda}_{\Delta{\bm{H}}} {\bm{V}}_{\widehat{\bm{H}}}^H\big). \end{align}

Based on (\ref{matrix_inequality_singular}), when the constraint functions $\psi_i(\bm{F})$'s in \textbf{Opt.\,4.1} are left unitarily invariant, the worst-case $ \Delta\bm{H}$ for \textbf{Opt.\,4.1} is
\begin{align}\label{worst_H}
   \Delta\bm{H}_{\rm worst}= {\bm{U}}_{\widehat{\bm{H}}}  \bm{\Lambda}_{\Delta{\bm{H}}} {\bm{V}}_{\widehat{\bm{H}}}^H.
\end{align} That is because when $\psi_i(\bm{F})$'s in \textbf{Opt.\,4.1} are left unitarily invariant, such as \textbf{Constraint 6:} the joint power constraint or \textbf{Constraint 5:} the constraints on the eigenvalues of $\bm{F}\bm{F}^{\rm{H}}$, for any feasible $\bm{F}$ and an arbitrary unitary matrix $\bm{Q}$, $\bm{Q}\bm{F}$ is also feasible. Thus it is always possible to find $\bm{Q}$ that makes the matrix inequality in (\ref{matrix_inequality_singular}) hold. Furthermore, based on the worst-case  $\Delta\bm{H}_{\rm worst}$ in (\ref{worst_H}), \textbf{Opt.\,4.1} is rewritten  as
 \begin{align}\label{eq78}
    \begin{array}{lcl}
 \hspace*{-5mm}\textbf{Opt.\,4.2:}  \hspace*{-3mm} & \max\limits_{\bm{F}} &  \hspace*{-3mm} \sigma_{\rm{n}}^{-2}\bm{F}^{\rm H}
 \big(\widehat{\bm{H}}\! -\! \Delta\bm{H}_{\rm worst}\big)^{\rm H}
 \big(\widehat{\bm{H}}\! -\! \Delta\bm{H}_{\rm worst}\big) \bm{F} \! , \\
 \!\!\! &\!\! {\rm{s.t.}}\!\! &  \hspace*{-2mm} \psi_i(\bm{F})\le 0 , 1\le i\le I,
 \end{array}
 \end{align}

 We would like to highlight that all the constraints $\psi_i(\bm{F})$'s in \textbf{Opt.\,4.1} are right unitarily invariant. The objective function in \textbf{Opt.\,4.2} is an upper bound of the worst case of the objective function of \textbf{Opt.\,4.1} and this bound is tight when $\psi_i(\bm{F})$'s are also left unitarily invariant. Specifically, when $\psi_i(\bm{F})$'s in \textbf{Opt.\,4.1} are right unitarily invariant, the term $\big(\widehat{\bm{H}}\! -\! \Delta\bm{H}_{\rm worst}\big)^{\rm H}
 \big(\widehat{\bm{H}}\! -\! \Delta\bm{H}_{\rm worst}\big)$ is the worst case of $\big(\widehat{\bm{H}}\! -\! \Delta\bm{H}\big)^{\rm H}
 \big(\widehat{\bm{H}}\! -\! \Delta\bm{H}\big)$ only when $\Delta\bm{H}$ is restricted to have the same SVD unitary matrices as
 $\widehat{\bm{H}}$.

{\emph{1)}~{\rm\bf Shaping Constraint}}:
 For the shaping constraint (\ref{eq46}), the optimal
 solution $\bm{F}_{\rm opt}$ of \textbf{Opt.\,4.2} is also specified by Lemma~\ref{L1}.
 That is, when the rank of $\bm{R}_{\rm s}$ is not larger than the number of columns
 and the number of rows in $\bm{F}$, $\bm{F}_{\rm opt}$ of \textbf{Opt.\,4.2} is a
 square root of $\bm{R}_{\rm s}$.

{\emph{2)}~{\rm\bf Joint Power Constraint}: Under the joint power constraint (\ref{eq47}), \textbf{Opt.\,4.2} is identical
 to \textbf{Opt.\,1.5} with
 \begin{align}
 \label{Pi_A}
 \bm{\Pi}=\sigma_{\rm{n}}^{-2}\big(\widehat{\bm{H}}\! -\! \Delta\bm{H}_{\rm worst}\big)^{\rm H}
 \big(\widehat{\bm{H}}\! -\! \Delta\bm{H}_{\rm worst}\big).
 \end{align} As a result, the Pareto optimal solutions
 $\bm{F}_{\rm opt}$ of \textbf{Opt.\,4.2} under the joint power constraint are defined
 exactly in Lemma~\ref{L2} by simply replacing $\bm{\Pi}$ in (\ref{eq32}) and (\ref{eq33})
 with (\ref{Pi_A}).


{\emph{3)}~{\rm\bf Multiple Weighted Power Constraints}:
 With the multiple weighted power constraints of (\ref{eq48}), the Pareto optimal solutions of
 \textbf{Opt.\,4.2}  are specified by Lemma~\ref{L3}, where
 $\bm{\Pi}$ should be replaced by (\ref{Pi_A}).

\section{Simulation Results and Discussions}\label{S9}

In this section, we take the MSE criterion (\textbf{Obj.\,2} in Table~\ref{tab0}) of MIMO systems as a central figure-of-merit to demonstrate the proposed  robust designs in Section~\ref{S3} with statistically  imperfect CSIT (ICSIT) and  CSIR (ICSIR), Section~\ref{S4} with statistically imperfect CSIT (ICSIT) and perfect CSIR (PCSIR), and Section~\ref{S5} with deterministically imperfect CSIT and CSIR.  Specifically,  in  Section~\ref{S3} and Section~\ref{S4},   the  sum average MSE is studied  to illustrate the influence of imperfect   CSIT  and/or CSIR on  average symbol detection performance.     While  in  Section~\ref{S5}, the worst-case MSE  is adopted  to  guarantee the symbol detection   performance  for all channel realizations
Notice that the  proposed  robust  designs  in above Sections all
have  analytical solutions, and can be reduced  to simple  power allocation  problems with water-filling solutions.

 Unless otherwise stated,  numerical  results are presented for  the point-to-point MIMO scenario with the transmitter and receiver equipped with $N_t=4$ and $N_r=4$ antennas, respectively. Moreover,  the number of data streams is $L=2$. According to  \eqref{Channel_Error}  adopted in Section~\ref{S3} and Section~\ref{S4},   we assume that the imperfect  CSI   consists  of the estimated  term   $\widehat{\bm{H}}$ distributed as $\mathcal{CN}(\bm{0}, (1-\sigma_e^2)\bm{I}_{N_r} \otimes \bm{I}_{N_t})$  and the error  term $\bm{H}_W\bm{\Psi}^{\frac{1}{2}}$,  where $\bm{\Psi}$  is defined  by the exponential model, i.e.,  $[\bm{\Psi}]_{m,n}=\sigma_e^2p_t^{\vert m-n\vert}$ with $p_t=0.5$ and $\sigma_e^2=0.1$, to  realize the normalized channel $\mathbb{E}\big\{\big[\bm{H}\big]_{m,n}
 \big[\bm{H}\big]_{m,n}^*\big\}=1$, $\forall m,n$. While for the worst-case optimization in Section~\ref{S5},  the  relative  error threshold subject to Spectral norm   is set as $\gamma={s}{\Vert \widehat{\bm{H}}\Vert_2}$ with $s\in [0,1]$. In addition,   for both  Section~\ref{S3} and Section~\ref{S4}, since the unknown  weighting factors need to be determined  via the sub-gradient method,  the per-antenna power  constraints, ${\rm{Tr}}({\bm \Omega}_i\bm{F}\bm{F}^H) \le P_{i}, \forall i=1, \cdots N_t$,  where ${\bm \Omega}_i=\text{diag}[\bm{0}_{1\times i-1}, 1, \bm{0}_{1\times N_t-i}], \forall i$ and  $P_i=P_t,\forall i$,    are  mainly studied  in the simulations as  a special case of  multiple weighted power
constraints.  While for Section~\ref{S5},   the
 joint power constraints $\text{Tr}(\bm{F}\bm{F}^H)\le (L-1){P}_t$ and $\bm{F}\bm{F}^H \preceq P_t\bm{I}_{N_t}$  are  investigated  due to the tractability.  Particularly,  the  per-antenna power limitation  $ \text{Tr}(\bm{\Omega}_i\bm{F}\bm{F}^H)\le P_t, \forall i$ can   be readily inferred  from  the  shaping  power constraint $\bm{F}\bm{F}^H\preceq P_t\bm{I}_{N_t}$.
We also define the SNR as  $\frac{P_t}{\sigma_n^2}$,
where $P_t=1$W  and the noise power $\sigma_n^2$ is   varied.

For a comprehensive comparison, we also consider three baselines for the MIMO scenario  as follows: For Section~\ref{S3} and Section~\ref{S4},   the naive design  that  simply regards  $\widehat{\bm{H}}$  as a perfect channel estimate of the  instantaneous channel ${\bm{H}}$  is  studied,  to which  the  optimal solution is derived  by solving the problem \eqref{eq54}/\eqref{eq70} with $\bm{\Psi}\!=\!\bm{0}$, and  the sum average MSE  is obtained through Monte Carlo experiments.
 While for
 Section~\ref{S5}, the  nonrobust    design  is studied by firstly considering  $\Delta\bm{H}=\bm{0}$ in the problem \eqref{eq71}, and   then  the obtained  optimal  solution is substituted into the inner minimization of the problem \eqref{eq71}   to find the worst-case MSE. Moreover, the ideal case  assuming both   PCSIR and perfect CSIT (PCSIT) is also considered for all above Sections.



Fig.~\ref{Fig4} shows the  sum average MSE of  all studied  designs  in Section~\ref{S3} and Section~\ref{S4} as the function of the  channel error $\sigma_e^2$. Clearly, it is observed  that when  $\sigma_e^2$ decreases, sum average MSE performances of all studied designs  improve.
Also,
the performance gap between  the  native design  and  robust design in Section~\ref{S3} with ICSIT and ICSIR    becomes narrowed. In particular, as the rise of $\sigma_e^2$,  the less performance loss  of the robust design in Section~\ref{S4} with  PCSIR compared to  that in Section~\ref{S3} with  ICSIR can be observed, which further indicates that PCSIR is crucial  to realize the  acceptable average  MSE performance.
\begin{figure}[!t]	\vspace*{-6mm}
	\begin{center}
		\includegraphics[width=0.4\textwidth]{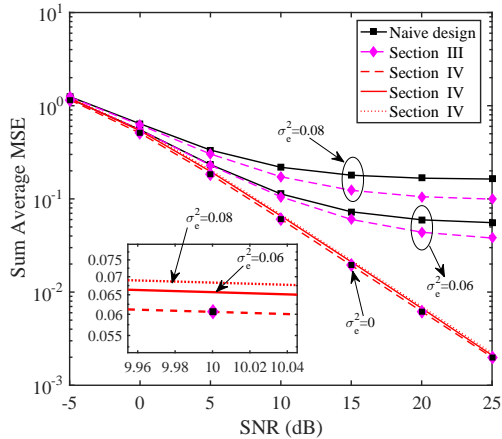}
	\end{center}
	\vspace*{-3mm}
	\caption{Sum average MSE versus the estimated channel error $\sigma_e^2$ for all studied designs in Section~\ref{S3} and Section~\ref{S4},.}
	\label{Fig4}
	\vspace*{-4mm}
\end{figure}

Fig.~\ref{Fig5} shows the  worst-case  MSE of  the proposed robust    design  in Section~\ref{S5} and the other baselines  as the function of SNR.  Naturally, the ideal design achieves the best worst-case MSE performance,  and   the proposed  robust design  is the next.  The nonrobust design  has the worst performance since the  robustness against channel error is not considered.  Similarly to the  robust design in Section~\ref{S4} with  ICSIT and PCSIR,  we also find that the  slopes of  all
 studied  worst-case  designs  are nearly identical,  and the corresponding  worst-case MSE performance  is  similar especially   at high SNR region, because that   in this context the high transmit power weakens
 the influence of deterministic channel error on  the achievable MSE.   The above  conclusions can also be drawn  when the increasing number of  date streams  $L=3$ is considered. In this context,    the performance gap among all studied  designs  is further enlarged.
\begin{figure}[!t]	\vspace*{-6mm}
	\begin{center}
		\includegraphics[width=0.4\textwidth]{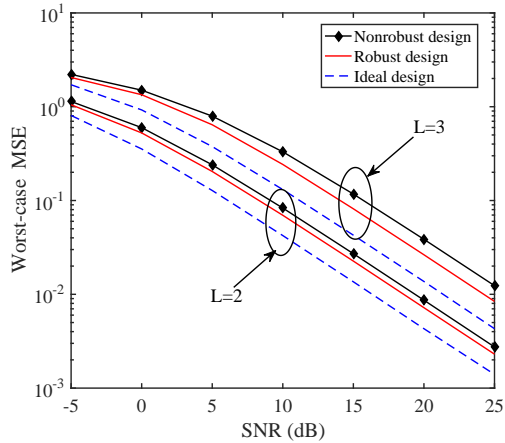}
	\end{center}
	\vspace*{-3mm}
	\caption{ Worst-case MSE versus SNR for the proposed     robust design in Section~\ref{S5} and all other baselines.}
	\label{Fig5}
	\vspace*{-4mm}
\end{figure}
\begin{figure}[!t]	\vspace*{-4mm}
	\begin{center}
		\includegraphics[width=0.4\textwidth]{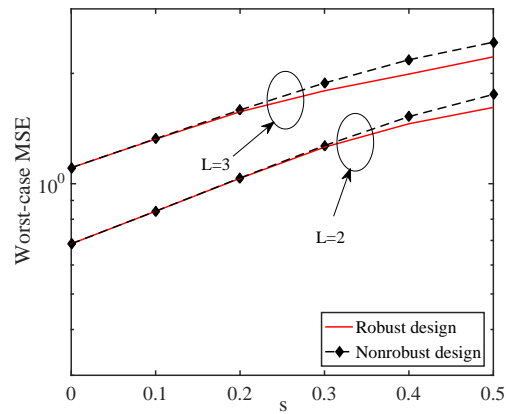}
	\end{center}
	\vspace*{-3mm}
	\caption{Worst-case MSE versus the relative  error threshold $s$ for the proposed     robust and nonrobust designs in Section~\ref{S5}.}
	\label{Fig7}
	\vspace*{-4mm}
\end{figure}

Fig.~\ref{Fig7} shows the  worst-case  MSE of  the proposed robust    design  in Section~\ref{S5} and the other baselines  as the function of the relative  error threshold $s$.  As expected, the robust design
outperforms the non-robust one, and the performance  gain becomes  more evident  with the increase of the  number of  date streams $L=3$ and error threshold $s$.

In order to assess the performance of the proposed solutions under
general multiple weighted power constraints, without loss of
generality we first build an exponential correlation matrix
$\bm{\Omega}$ with $[\bm{\Omega}]_{i,j}=0.3^{|i-j|}$. Based on
$\bm{\Omega}$, a pair of weighting matrices i.e., $\bm{\Omega}_1$
and $\bm{\Omega}_2$ are constructed. Specifically, $\bm{\Omega}_1$
corresponds to the first two eigenchannels and $\bm{\Omega}_2$
corresponds to the last two eigenchannels.  In addition, the power
ratio between the two constraints is 0.6 and 0.4. Moreover, in the
simulations the numbers of antennas and data streams are equal to each
other. Then both the MSE minimization and sum rate maximization are
convex, hence the problem can be solved by using CVX \cite{tool}.  It
can be concluded from Fig.~\ref{Fig801} and Fig.~\ref{Fig802} that the proposed solutions
have the same performance as that computed by CVX for all the settings
investigated. We would like to point out that the numerical algorithms
based on CVX have no tangible physical meanings and suffer from high
computational complexity as well as from limited scalability.

\section{Conclusions}\label{S10}

 In this paper, a comprehensive framework for matrix-monotonic optimization has been given
 under various power constraints, including shaping constraint, joint power constraint and
 multiple weighted power constraints. Matrix-monotonic optimization of three
 different CSI scenarios have been investigated in depth, which are: 1)~both
 transmitter and receiver have imperfect CSI; 2)~perfect CSI is available at the receiver
 but the transmitter has only channel statistics; and 3)~perfect CSI is available at the receiver, but the channel
 estimation error at the transmitter is norm-bounded. In all three cases, the matrix-monotonic optimization framework has been used to derive closed-form optimal
 structures of the optimal matrix variables, which significantly simplifies the associated
 optimization problems and reveals a range of underlying physical insights.
\begin{figure}[!t]	\vspace*{-8mm}
	\begin{center}
		\includegraphics[width=0.4\textwidth]{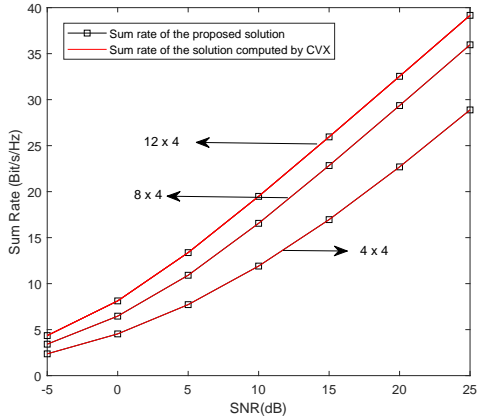}
	\end{center}
	\vspace*{-3mm}
	\caption{The performance comparisons between the proposed solution and the solution computed by CVX in terms of sum rate under multiple weighted power constraints.}
	\label{Fig801}
	\vspace*{-4mm}
\end{figure}
\begin{figure}[!t]	\vspace*{-8mm}
	\begin{center}
		\includegraphics[width=0.4\textwidth]{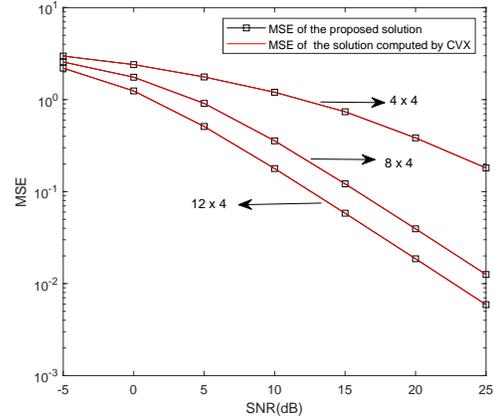}
	\end{center}
	\vspace*{-3mm}
	\caption{The performance comparisons between the proposed solution and the solution computed by CVX in terms of sum MSE under multiple weighted power constraints.}
	\label{Fig802}
	\vspace*{-4mm}
\end{figure}

\appendices

\section{Proof of Table~\ref{tab1}}
\label{appendix_table_II}

In the following, the detailed proofs for the optimal $\bm{Q}_{\bm{X}}$ in Table~\ref{tab1} are given.

\noindent  \textbf{Obj. 1:}  For \textbf{Obj. 1}, based on the righthand side of Matrix Inequality 3 and taking $\bm{C}=\bm{Q}^{\rm{H}}_{\bm{X}}\bm{F}^{\rm{H}}{\boldsymbol{\Pi}}\bm{F}\bm{Q}_{\bm{X}}$
and $\bm{D}={\boldsymbol \Phi}$, the optimal $\bm{Q}_{\bm{X}}$ can be derived.

\noindent  \textbf{Obj. 2:}  For \textbf{Obj. 2}, based on the lefthand side of Matrix Inequality 2 and taking $\bm{C}=\bm{Q}^{\rm{H}}_{\bm{X}}\bm{F}^{\rm{H}}{\boldsymbol{\Pi}}\bm{F}\bm{Q}_{\bm{X}}$
and $\bm{D}={\boldsymbol \Phi}$, the optimal $\bm{Q}_{\bm{X}}$ can be derived.

\noindent  \textbf{Obj. 3:}  For \textbf{Obj. 3}, based on the lefthand side of Matrix Inequality 1 and taking $\bm{C}=(\bm{Q}^{\rm{H}}_{\bm{X}}\bm{F}^{\rm{H}}{\boldsymbol{\Pi}}\bm{F}
\bm{Q}_{\bm{X}}+\alpha\bm{I})^{-1}$
and $\bm{D}=\bm{A}\bm{A}^{\rm{H}}$, the optimal $\bm{Q}_{\bm{X}}$ can be derived.

\noindent  \textbf{Obj. 4:} For \textbf{Obj. 4}, the following equality holds
\begin{align}
&{\rm{log}}|\bm{A}^{\rm{H}}(\bm{Q}^{\rm{H}}_{\bm{X}}\bm{F}^{\rm{H}}{\boldsymbol{\Pi}}\bm{F}
\bm{Q}_{\bm{X}}+\alpha\bm{I})^{-1}\bm{A}+\boldsymbol{\Phi}| \nonumber \\
=& \ {\rm{log}}|\bm{A}^{\rm{H}}(\bm{Q}^{\rm{H}}_{\bm{X}}\bm{F}^{\rm{H}}{\boldsymbol{\Pi}}\bm{F}
\bm{Q}_{\bm{X}}+\alpha\bm{I})^{-1}\bm{A}\boldsymbol{\Phi}^{-1}+\bm{I}|+c \nonumber \\
=& \ {\rm{log}}|(\bm{Q}^{\rm{H}}_{\bm{X}}\bm{F}^{\rm{H}}{\boldsymbol{\Pi}}\bm{F}
\bm{Q}_{\bm{X}}+\alpha\bm{I})^{-1}\bm{A}\boldsymbol{\Phi}^{-1}\bm{A}^{\rm{H}}+\bm{I}|+c,
\end{align}where $c={\rm{log}}|\boldsymbol{\Phi}|$ is a constant. Then
based on the lefthand side of Matrix Inequality 4 and taking $\bm{C}=(\bm{Q}^{\rm{H}}_{\bm{X}}\bm{F}^{\rm{H}}{\boldsymbol{\Pi}}\bm{F}
\bm{Q}_{\bm{X}}+\alpha\bm{I})^{-1}$
and $\bm{D}=\bm{A}\boldsymbol{\Phi}^{-1}\bm{A}^{\rm{H}}$, the optimal $\bm{Q}_{\bm{X}}$ can be derived.

\noindent \textbf{Obj. 5.1:} For \textbf{Obj. 5.1} based on the majorization theory the minimum value of an additively Schur-convex function is achieved when all the diagonal elements of $(\bm{Q}^{\rm{H}}_{\bm{X}}\bm{F}^{\rm{H}}{\boldsymbol{\Pi}}\bm{F}
\bm{Q}_{\bm{X}}+\alpha\bm{I})^{-1}$ equal with each other. Therefore the optimal $\bm{Q}_{\bm{X}}$ can be derived.

\noindent \textbf{Obj. 5.2:} For \textbf{Obj. 5.2} based on the majorization theory the minimum value of an additively Schur-convex function is achieved when all the diagonal elements of $(\bm{Q}^{\rm{H}}_{\bm{X}}\bm{F}^{\rm{H}}{\boldsymbol{\Pi}}\bm{F}
\bm{Q}_{\bm{X}}+\alpha\bm{I})^{-1}$ equals the eigenvalues of $(\bm{Q}^{\rm{H}}_{\bm{X}}\bm{F}^{\rm{H}}{\boldsymbol{\Pi}}\bm{F}
\bm{Q}_{\bm{X}}+\alpha\bm{I})^{-1}$. Therefore the optimal $\bm{Q}_{\bm{X}}$ can be derived.

\noindent \textbf{Obj. 6.1:} For \textbf{Obj. 6.1}  based on the majorization theory the minimum value of a multiplicatively Schur-convex function is achieved when all the diagonal elements of the lower triangular matrix of the Cholesky-decomposition $(\bm{Q}^{\rm{H}}_{\bm{X}}\bm{F}^{\rm{H}}{\boldsymbol{\Pi}}\bm{F}
\bm{Q}_{\bm{X}}+\alpha\bm{I})^{-1}$ have the same squared modulus. Therefore, the optimal $\bm{Q}_{\bm{X}}$ can be derived.

\noindent \textbf{Obj. 6.2:} For \textbf{Obj. 6.2} based on the majorization theory the minimum value of a multiplicatively Schur-convex function is achieved when all the squared diagonal elements of the lower triangular matrix of the Cholesky-decomposition $(\bm{Q}^{\rm{H}}_{\bm{X}}\bm{F}^{\rm{H}}{\boldsymbol{\Pi}}\bm{F}
\bm{Q}_{\bm{X}}+\alpha\bm{I})^{-1}$ equal the eigenvalues of $(\bm{Q}^{\rm{H}}_{\bm{X}}\bm{F}^{\rm{H}}{\boldsymbol{\Pi}}\bm{F}
\bm{Q}_{\bm{X}}+\alpha\bm{I})^{-1}$. Therefore the optimal $\bm{Q}_{\bm{X}}$ can be derived.

\noindent \textbf{Obj. 7:} For \textbf{Obj. 7} the objective function can be rewritten as
\begin{align}
& {\rm{log}}|\bm{A}^{\rm{H}}\bm{Q}^{\rm{H}}_{\bm{X}}\bm{F}^{\rm{H}}{\boldsymbol{\Pi}}\bm{F}
\bm{Q}_{\bm{X}}\bm{A}+{\boldsymbol \Phi}| \nonumber \\
= \ & {\rm{log}}|\bm{A}^{\rm{H}}\bm{Q}^{\rm{H}}_{\bm{X}}\bm{F}^{\rm{H}}{\boldsymbol{\Pi}}\bm{F}
\bm{Q}_{\bm{X}}\bm{A}{\boldsymbol \Phi}^{-1}+\bm{I}|+c \nonumber \\
= \ & {\rm{log}}|\bm{Q}^{\rm{H}}_{\bm{X}}\bm{F}^{\rm{H}}{\boldsymbol{\Pi}}\bm{F}
\bm{Q}_{\bm{X}}\bm{A}{\boldsymbol \Phi}^{-1}\bm{A}^{\rm{H}}+\bm{I}|+c
\end{align}where $c={\rm{log}}|\boldsymbol{\Phi}|$.  Then
based on the righthand side of Matrix Inequality 4 and taking $\bm{C}=\bm{Q}^{\rm{H}}_{\bm{X}}\bm{F}^{\rm{H}}{\boldsymbol{\Pi}}\bm{F}
\bm{Q}_{\bm{X}}$
and $\bm{D}=\bm{A}\boldsymbol{\Phi}^{-1}\bm{A}^{\rm{H}}$ the optimal $\bm{Q}_{\bm{X}}$ can be derived.

\noindent \textbf{Obj. 8:} For \textbf{Obj. 8} the objective function equals
\begin{align}
& {\rm{Tr}}\left( (\bm{A}^{\rm{H}}\bm{Q}^{\rm{H}}_{\bm{X}}\bm{F}^{\rm{H}}{\boldsymbol{\Pi}}\bm{F}
\bm{Q}_{\bm{X}}\bm{A}+\alpha\bm{I})^{-1} \right) \nonumber \\
= \ & \sum_k \frac{1}{\lambda_k(\bm{A}^{\rm{H}}\bm{Q}^{\rm{H}}_{\bm{X}}\bm{F}^{\rm{H}}{\boldsymbol{\Pi}}\bm{F}
\bm{Q}_{\bm{X}}\bm{A})+\alpha}.
\end{align} At high SNR, when $k\le N$ and $N$ is the rank of $\bm{A}^{\rm{H}}\bm{Q}^{\rm{H}}_{\bm{X}}\bm{F}^{\rm{H}}{\boldsymbol{\Pi}}\bm{F}
\bm{Q}_{\bm{X}}\bm{A}$ we have
\begin{align}
\lambda_k(\bm{A}^{\rm{H}}\bm{Q}^{\rm{H}}_{\bm{X}}\bm{F}^{\rm{H}}{\boldsymbol{\Pi}}\bm{F}
\bm{Q}_{\bm{X}}\bm{A})\gg \alpha.
\end{align} Then at high SNR, minimizing  \textbf{Obj. 8} is equivalent to minimizing the following function
\begin{align}
& \sum_{k=1}^N \frac{1}{\lambda_k(\bm{A}^{\rm{H}}\bm{Q}^{\rm{H}}_{\bm{X}}\bm{F}^{\rm{H}}{\boldsymbol{\Pi}}\bm{F}
\bm{Q}_{\bm{X}}\bm{A})+\alpha} \nonumber \\
\approx & \sum_{k=1}^N \frac{1}{\lambda_k(\bm{A}^{\rm{H}}\bm{Q}^{\rm{H}}_{\bm{X}}\bm{F}^{\rm{H}}{\boldsymbol{\Pi}}\bm{F}
\bm{Q}_{\bm{X}}\bm{A})}.
\end{align} Because $\sum_k {1}/{x}$ is multiplicatively Schur-convex, based on majorization theory the optimal $\bm{Q}_{\bm{X}}$ can be derived.

\noindent \textbf{Obj. 9:} For \textbf{Obj. 9} at high SNR the objective function can be approximated by
\begin{align}
& {\rm{Tr}}\left(\bm{A}^{\rm{H}}(\bm{Q}^{\rm{H}}_{\bm{X}}\bm{F}^{\rm{H}}{\boldsymbol{\Pi}}\bm{F}
\bm{Q}_{\bm{X}}+{\boldsymbol \Phi})^{-1}\bm{A}\right) \nonumber \\
\approx \ & {\rm{Tr}}\left(\bm{A}^{\rm{H}}(\bm{Q}^{\rm{H}}_{\bm{X}}\bm{F}^{\rm{H}}{\boldsymbol{\Pi}}\bm{F}
\bm{Q}_{\bm{X}})^{-1}\bm{A}\right),
\end{align} where the approximation comes from the fact that at high SNR $\bm{Q}^{\rm{H}}_{\bm{X}}\bm{F}^{\rm{H}}{\boldsymbol{\Pi}}\bm{F}
\bm{Q}_{\bm{X}}$ is much larger than ${\boldsymbol \Phi}$ and $\bm{Q}^{\rm{H}}_{\bm{X}}\bm{F}^{\rm{H}}{\boldsymbol{\Pi}}\bm{F}
\bm{Q}_{\bm{X}}+{\boldsymbol \Phi}\approx\bm{Q}^{\rm{H}}_{\bm{X}}\bm{F}^{\rm{H}}{\boldsymbol{\Pi}}\bm{F}
\bm{Q}_{\bm{X}}$.
Then based on the lefthand side of Matrix Inequality 1 and taking $\bm{C}=(\bm{Q}^{\rm{H}}_{\bm{X}}\bm{F}^{\rm{H}}{\boldsymbol{\Pi}}\bm{F}
\bm{Q}_{\bm{X}})^{-1}$
and $\bm{D}=\bm{A}\bm{A}^{\rm{H}}$, the optimal $\bm{Q}_{\bm{X}}$ can be derived.

\noindent \textbf{Obj. 10:}  For \textbf{Obj. 10} based on the righthand side of Matrix Inequality 3 and taking $\bm{C}=\bm{Q}^{\rm{H}}_{\bm{X}}\bm{F}^{\rm{H}}{\boldsymbol{\Pi}}\bm{F}
\bm{Q}_{\bm{X}}\otimes{\boldsymbol\Sigma}_2$
and $\bm{D}=\boldsymbol{\Phi}\otimes{\boldsymbol\Sigma}_1$ the optimal $\bm{Q}_{\bm{X}}$ can be derived.

\noindent \textbf{Obj. 11:}  For \textbf{Obj. 11} based on the righthand side of Matrix Inequality 3 and taking $\bm{C}={\boldsymbol\Sigma}_2\otimes\bm{Q}^{\rm{H}}_{\bm{X}}\bm{F}^{\rm{H}}{\boldsymbol{\Pi}}\bm{F}
\bm{Q}_{\bm{X}}$
and $\bm{D}={\boldsymbol\Sigma}_1\otimes\boldsymbol{\Phi}$ the optimal $\bm{Q}_{\bm{X}}$ can be derived.

\noindent \textbf{Obj. 12:}  For \textbf{Obj. 12} based on the lefthand side of Matrix Inequality 2 and taking $\bm{C}=\bm{Q}^{\rm{H}}_{\bm{X}}\bm{F}^{\rm{H}}{\boldsymbol{\Pi}}\bm{F}
\bm{Q}_{\bm{X}}\otimes{\boldsymbol\Sigma}_2$
and $\bm{D}=\boldsymbol{\Phi}\otimes{\boldsymbol\Sigma}_1$ the optimal $\bm{Q}_{\bm{X}}$ can be derived.

\noindent \textbf{Obj. 13:}  For \textbf{Obj. 13} based on the righthand side of Matrix Inequality 2 and taking $\bm{C}={\boldsymbol\Sigma}_2\otimes\bm{Q}^{\rm{H}}_{\bm{X}}\bm{F}^{\rm{H}}{\boldsymbol{\Pi}}\bm{F}
\bm{Q}_{\bm{X}}$
and $\bm{D}={\boldsymbol\Sigma}_1\otimes\boldsymbol{\Phi}$, the optimal $\bm{Q}_{\bm{X}}$ can be derived.

\noindent \textbf{Obj. 14:}  For \textbf{Obj. 14} based on the lefthand side of Matrix Inequality 1 and taking $\bm{C}=(\bm{I}+\bm{Q}^{\rm{H}}_{\bm{X}}\bm{F}^{\rm{H}}{\boldsymbol{\Pi}}\bm{F}
\bm{Q}_{\bm{X}}\otimes{\boldsymbol\Sigma}_2)^{-1}$
and $\bm{D}=(\bm{A}\bm{A}^{\rm{H}})\otimes{\boldsymbol\Sigma}_1$, the optimal $\bm{Q}_{\bm{X}}$ can be derived.

\noindent \textbf{Obj. 15:}  For \textbf{Obj. 15} based on the righthand side of Matrix Inequality 2 and taking $\bm{C}=(\bm{I}+{\boldsymbol\Sigma}_2\otimes\bm{Q}^{\rm{H}}_{\bm{X}}\bm{F}^{\rm{H}}{\boldsymbol{\Pi}}\bm{F}
\bm{Q}_{\bm{X}})^{-1}$
and $\bm{D}={\boldsymbol\Sigma}_1\otimes(\bm{A}\bm{A}^{\rm{H}})$, the optimal $\bm{Q}_{\bm{X}}$ can be derived.

\section{Proof of Lemma~\ref{L2}}\label{appendix_1}

The Pareto optimal solution set of (\ref{eq31}) falls in the optimal solution set of the
 following optimization problem for all the possible ${\bm{F}}_{\rm in}$ that are
 in the region of ${\rm Tr}\big({\bm{F}}{\bm{F}}^{\rm H}\big)\le P$ and $\bm{F}\bm{F}^{\rm H}\preceq \tau \bm{I}$:
\begin{align}
\label{equ_app_1}
\max_{\bm{F},\alpha} & \ \ \ \ \ \ \ \alpha  \nonumber \\
{\rm s.t.} \ \ &\bm{F}^{\rm H}\bm{\Pi}\bm{F}=\alpha \bm{F}^{\rm H}_{\rm{in}}\bm{\Pi}\bm{F}_{\rm{in}} \nonumber \\
 & {\rm Tr}\big(\bm{F}\bm{F}^{\rm H}\big) \hspace{-1mm}\le\hspace{-1mm} P , \bm{F}\bm{F}^{\rm H}\preceq \tau \bm{I} .
\end{align} This conclusion is obvious because for any Pareto optimal solution of (\ref{eq31}) $\bm{F}_{\rm{opt}}$, for an $\alpha<1$, $\bm{F}_{\rm{in}}=\alpha\bm{F}_{\rm{opt}}$ obviously satisfies ${\rm Tr}\big({\bm{F}}_{\rm{in}}{\bm{F}}_{\rm{in}}^{\rm H}\big)\le P$ and $\bm{F}_{\rm{in}}\bm{F}_{\rm{in}}^{\rm H}\preceq \tau \bm{I}$. In the following, we will prove that the optimal solutions of (\ref{equ_app_1}) own the same structure.

Based on the matrix equality properties that when $\bm{A}$ and $\bm{B}$ have the same dimensionality $\bm{A}^{\rm{H}}{\bm A}=\bm{B}^{\rm{H}}{\bm B}$ is equivalent to $\bm{A}=\bm{U}\bm{A}$ with $\bm{U}$ being an unitary matrix \cite[P406]{Horn1990}, the first constraint of (\ref{equ_app_1}) is equivalent to
\begin{align}
\bm{\Pi}^{1/2}\bm{F}=\sqrt{\alpha} \bm{U}\bm{\Pi}^{1/2}\bm{F}_{\rm{in}}
\end{align}based on which and defining the pseudo inverse of $\bm{\Pi}^{1/2}$ as $(\bm{\Pi}^{1/2})^{\dag}$, we have
\begin{align}
(\bm{\Pi}^{1/2})^{\dag}\bm{\Pi}^{1/2}\bm{F}=\sqrt{\alpha} (\bm{\Pi}^{1/2})^{\dag}\bm{U}\bm{\Pi}^{1/2}\bm{F}_{\rm{in}}.
\end{align} It is obvious that $\alpha$ is solved to be
\begin{align}
\alpha=\frac{{\rm{Tr}}\left[\left((\bm{\Pi}^{1/2})^{\dag}\bm{\Pi}^{1/2}\bm{F}\right)^{\rm{H}}(\bm{\Pi}^{1/2})^{\dag}\bm{\Pi}^{1/2}\bm{F}\right]}
{{\rm{Tr}}\left[\left((\bm{\Pi}^{1/2})^{\dag}\bm{U}\bm{\Pi}^{1/2}\bm{F}_{\rm{in}}\right)^{\rm{H}}(\bm{\Pi}^{1/2})^{\dag}\bm{U}\bm{\Pi}^{1/2}\bm{F}_{\rm{in}}\right]}.
\end{align} Based on \textbf{Matrix Inequality 1}, the numerator of the righthand side of the above equation satisfies
\begin{align}
\label{equ_app_2}
{\rm{Tr}}\left[((\bm{\Pi}^{1/2})^{\dag}\bm{\Pi}^{1/2}\bm{F})^{\rm{H}}(\bm{\Pi}^{1/2})^{\dag}\bm{\Pi}^{1/2}\bm{F}\right]\le \sum_j \lambda_j({\bm{F}}{\bm{F}}^{\rm{H}})
\end{align}  and meanwhile the denominator satisfies
\begin{align}
\label{equ_app_3}
&{\rm{Tr}}\left[\left((\bm{\Pi}^{1/2})^{\dag}\bm{U}\bm{\Pi}^{1/2}\bm{F}_{\rm{in}}\right)^{\rm{H}}(\bm{\Pi}^{1/2})^{\dag}\bm{U}\bm{\Pi}^{1/2}\bm{F}_{\rm{in}}
\right]
\nonumber \\
\ge&
\sum_j
   \frac{\lambda_j\big(\bm{\Pi}^{1/2}{\bm{F}}_{\rm in}{\bm{F}}_{\rm in}^{\rm{H}}\bm{\Pi}^{1/2}\big)}
  {\lambda_j(\bm{\Pi}\big)}.
\end{align} Based on (\ref{equ_app_2}) and (\ref{equ_app_3}), $\alpha$ is maximized when $\bm{F}$ satisfies the following structure
\begin{align}\label{equ_app_4}
\bm{F}={\bm{U}}_{{\bm{\Pi}}}{\bm{\Lambda}}_{\bm{F}}{\bm{V}}_{\rm{in}}^{\rm{H}}
\end{align} where the unitary matrices ${\bm{U}}_{{\bm{\Pi}}}$ and ${\bm{V}}_{\rm{in}}$ are defined based on the following EVDs
\begin{align}
&{\bm{\Pi}}={\bm{U}}_{{\bm{\Pi}}}{\bm{\Lambda}}_{\bm{\Pi}}{\bm{U}}_{{\bm{\Pi}}}^{\rm{H}} \ {\rm{with}} \ {\bm{\Lambda}}_{\bm{\Pi}} \searrow, \nonumber \\
&\bm{F}^{\rm H}_{\rm{in}}\bm{\Pi}\bm{F}_{\rm{in}} ={\bm{V}}_{\rm{in}}{\bm{\Lambda}}_{\rm{in}}{\bm{V}}_{\rm{in}}^{\rm{H}} \ {\rm{with}} \ {\bm{\Lambda}}_{\rm{in}} \searrow.
\end{align}

It is worth noting that the final two constraints in the optimization problem (\ref{equ_app_1}) only constrain the eigenvalues of $\bm{F}\bm{F}^{\rm{H}}$. In other words, the final two constraints in (\ref{equ_app_1}) only constrain the singular values of $\bm{F}$.
Moreover, it is obvious that the final two constraints in (\ref{equ_app_1}) are both right unitarily invariant and left unitarily invariant.
The derivations in (\ref{equ_app_2}) and (\ref{equ_app_3}) are independent of the singular values of of $\bm{F}$. It means that for any given $\bm{\Lambda}_{\bm{F}}$ in (\ref{equ_app_4}) the optimal $\bm{F}$ maximizing $\alpha$ satisfies the structure in (\ref{equ_app_4}) without violating the final two constraints in (\ref{equ_app_1}). Therefore, it is concluded that the optimal solutions of (\ref{equ_app_1}) satisfies the structure in (\ref{equ_app_4}) and thus the Pareto optimal solutions of (\ref{eq53}) satisfies the structure given by (\ref{equ_app_4}). Furthermore, substituting (\ref{equ_app_4}) into the objective function of (\ref{eq53}), it can be seen that for the Pareto optimal solutions, the value of the unitary matrix ${\bm{V}}_{\rm{in}}$  in (\ref{equ_app_4}) does not affect the optimality of the Paremto optimal solutions. Finally the Pareto optimal solutions of (\ref{eq53}) satisfies the following structure
\begin{align}
\bm{F}_{\rm opt} =& \bm{U}_{\bm{\Pi}}\bm{\Lambda}_{\bm{F}}\bm{U}_{\text{Arb}}^{\rm H}.
\end{align}

\section{Proof of Convexity}\label{appendix_convexity}

The proof is exactly based the definition of convex function. Taking the complex matrix $\bm{F}$ as optimization variable, based on the definition of convex function, the function ${\rm Tr}\big(\bm{\Omega}_i\bm{F}\bm{F}^{\rm H}\big)$ is convex with respect to $\bm{F}$ if and only if for two complex matrices $\bm{F}_1$ and $\bm{F}_2$ and $0\le a\le1$ the following inequality holds \cite{Boyd04}
\begin{align}
&{\rm Tr}\big(\bm{\Omega}_i\bm{F}_1\bm{F}_1^{\rm H}\big)a+{\rm Tr}\big(\bm{\Omega}_i\bm{F}\bm{F}^{\rm H}\big)(1-a) \ge \nonumber \\
& {\rm Tr}\big(\bm{\Omega}_i(a\bm{F}_1+(1-a)\bm{F}_2)(a\bm{F}_1+(1-a)\bm{F}_2)^{\rm H}\big).
\end{align} This inequality can be proved when $\bm{\Omega}_i$ is a positive semidefinite matrix because
\begin{align}
& {\rm Tr}\big(\bm{\Omega}_i\bm{F}_1\bm{F}_1^{\rm H}\big)a+{\rm Tr}\big(\bm{\Omega}_i\bm{F}\bm{F}^{\rm H}\big)(1-a)\nonumber \\
&-
{\rm Tr}\big(\bm{\Omega}_i(a\bm{F}_1+(1-a)\bm{F}_2)(a\bm{F}_1+(1-a)\bm{F}_2)^{\rm H}\big)\nonumber \\
=&a(1-a){\rm Tr}\big(\bm{\Omega}_i(\bm{F}_1-\bm{F}_2)(\bm{F}_1-\bm{F}_2)^{\rm H}\big)\ge0.
\end{align} Therefore, it can be concluded that  when $\bm{\Omega}_i$ is a positive semidefinite matrix ${\rm Tr}\big(\bm{\Omega}_i\bm{F}\bm{F}^{\rm H}\big)$ is convex with respect to $\bm{F}$.

\section{Proof of Lemma~\ref{L3}}\label{Apa}

 Any Pareto optimal solution of \textbf{Opt.\,1.6}, $\bm{F}_{\text{Pareto}}$, is also a
 Pareto optimal solution of the following multi-objective optimization problem
\begin{align}\label{App_1} 
 \min_{\bm{F}} \ & \big\{{\rm Tr}\big(\bm{\Omega}_i\bm{F}\bm{F}^{\rm H}\big)\big\}_{i=1}^I,
 {\rm{s.t.}}  \bm{F}^{\rm H}\bm{\Pi}\bm{F}=\bm{F}_{\text{Pareto}}^{\rm H}\bm{\Pi}\bm{F}_{\text{Pareto}}.
\end{align} This transformation is built on the proof by contradiction. If  $\bm{F}_{\text{Pareto}}$ is not a  Pareto optimal solution of (\ref{App_1}), it means that we can find a matrix $\bm{F}_1$ satisfying
\begin{align}
& \bm{F}_1^{\rm H}\bm{\Pi}\bm{F}_1=\bm{F}_{\text{Pareto}}^{\rm H}\bm{\Pi}\bm{F}_{\text{Pareto}}, \nonumber \\
& {\rm Tr}\big(\bm{\Omega}_i\bm{F}_1\bm{F}_1^{\rm H}\big) \le  {\rm Tr}\big(\bm{\Omega}_i\bm{F}_{\text{Pareto}}\bm{F}_{\text{Pareto}}^{\rm H}\big)\le P_i.
\end{align} Moreover, at least there  exits an index $i$ for which the first inequality in the second line will hold.
Then it is obvious that we can find a matrix $\bm{F}_2$ satisfying
\begin{align}
& {\rm Tr}\big(\bm{\Omega}_i\bm{F}_1\bm{F}_1^{\rm H}\big)\le {\rm Tr}\big(\bm{\Omega}_i\bm{F}_2\bm{F}_2^{\rm H}\big)\le P_i, \nonumber \\
 & \bm{F}_2^{\rm H}\bm{\Pi}\bm{F}_2\succeq \bm{F}_1^{\rm H}\bm{\Pi}\bm{F}_1= \bm{F}_{\text{Pareto}}^{\rm H}\bm{\Pi}\bm{F}_{\text{Pareto}}.
\end{align} As a result, a contradiction is achieved on the assumption that $\bm{F}_{\text{Pareto}}$ is Pareto optimal. 

 Since the constraint of (\ref{App_1}) is equivalent to $\bm{U}\bm{\Pi}^{\frac{1}{2}}\bm{F}=
 \bm{\Pi}^{\frac{1}{2}}\bm{F}_{\rm Pareto}$, where $\bm{U}$ is a suitable unitary matrix \cite[P406]{Horn1990}, the
 optimization problem (\ref{App_1}) is equivalent to
\begin{align}\label{App_2}
 \min\limits_{\bm{F}} \big\{{\rm Tr}\big(\bm{\Omega}_i\bm{F}\bm{F}^{\rm H}\big)\big\}_{i=1}^I , \
 {\rm{s.t.}} \bm{U}\bm{\Pi}^{\frac{1}{2}}\bm{F}=\bm{\Pi}^{\frac{1}{2}}\bm{F}_{\text{Pareto}} .
\end{align}
 In (\ref{App_2}), the objective functions are quadratic functions and the constraint
 is a linear function with respect to $\bm{F}$, which means that the multi-objective
 optimization problem (\ref{App_2}) is convex \cite[P135]{Boyd04} and the corresponding proof is given in Appendix \ref{appendix_convexity}.  Therefore, for any Pareto
 optimal solution of (\ref{App_2}), there exist the weights $\alpha_i$, $1\le i\le I$,
 for ensuring that the Pareto optimal solution can be computed via solving the following weighted
 sum optimization problem \cite[P179]{Boyd04}
\begin{align}\label{ap3}
 \min\limits_{\bm{F}} {\sum}_{i=1}^I \alpha_i {\rm Tr}\big(\bm{\Omega}_i\bm{F}\bm{F}^{\rm H}\big) , \
 {\rm{s.t.}} \bm{U}\bm{\Pi}^{\frac{1}{2}}\bm{F}=\bm{\Pi}^{\frac{1}{2}}\bm{F}_{\text{Pareto}} .
\end{align}The above conclusion for computing Pareto optimal solution of (\ref{App_2}) using weights $\alpha_i$, $1\le i\le I$ in (\ref{ap3}) are feasible to any unitary matrix $\bm{U}$. Meanwhile, it is worth noting that $\bm{F}_{\text{Pareto}}^{\rm H}\bm{\Pi}\bm{F}_{\text{Pareto}}$ is equivalent to $\bm{U}\bm{\Pi}^{\frac{1}{2}}\bm{F}=
 \bm{\Pi}^{\frac{1}{2}}\bm{F}_{\rm Pareto}$ where $\bm{U}$ is a suitable unitary matrix \cite[P406]{Horn1990}.
Thus the whole Pareto optimal solution set of (\ref{App_1}) can be achieved via solving
 the following optimization problem by changing the weights $\alpha_i$, $1\le i\le I$,
\begin{align}\label{ap4}
 \min\limits_{\bm{F}} & {\sum}_{i=1}^I\hspace*{-2mm} \alpha_i {\rm Tr}\big(\bm{\Omega}_i\bm{F}\bm{F}^{\rm H}\big) , 
 {\rm{s.t.}} \bm{F}^{\rm H}\bm{\Pi}\bm{F}\hspace{-1mm}=\hspace{-1mm}\bm{F}_{\text{Pareto}}^{\rm H}\bm{\Pi}\bm{F}_{\text{Pareto}} .
\end{align} For the optimal solution of (\ref{ap4}), ${\sum}_{i=1}^I\hspace*{-1mm} \alpha_i {\rm Tr}\big(\bm{\Omega}_i\bm{F}\bm{F}^{\rm H}\big)\hspace*{-1mm} =\hspace*{-1mm} P$.
 Then $\bm{F}_{\text{Pareto}}$ is a Pareto optimal solution of the following optimization problem
\begin{align}\label{label_opt} 
\hspace*{-3mm} \max_{\bm{F}} \bm{F}^{\rm H}\bm{\Pi}\bm{F} , \
 {\rm{s.t.}} {\rm Tr}\Big( {\sum}_{i=1}^I \alpha_i\bm{\Omega}_i\bm{F}\bm{F}^{\rm H}\Big)
  \le P.
\end{align}This is concluded based on the proof of contradiction. If $\bm{F}_{\text{Pareto}}$ is not a Pareto optimal solution of (\ref{label_opt}), for (\ref{label_opt}) there will exist $\bm{F}_1$ which satisfies $\bm{F}_1^{\rm H}\bm{\Pi}\bm{F}_1\hspace{-1mm}\succeq \hspace{-1mm}\bm{F}_{\text{Pareto}}^{\rm H}\bm{\Pi}\bm{F}_{\text{Pareto}}$ and $ {\rm Tr}\Big( {\sum}_{i=1}^I \alpha_i\bm{\Omega}_i\bm{F}_1\bm{F}_1^{\rm H}\Big)
  \le P$. For positive semidefinite matrices, $\bm{A}\succeq \bm{B}$  implies ${\bm{\lambda}}(\bm{A})\succeq {\bm{\lambda}}(\bm{B})$ \cite[P471]{Horn1990}. Meanwhile, for two complex matrices  $\bm{C}$ and $\bm{D}$, $\bm{C}\bm{D}$ and $\bm{D}\bm{C}$ have the same nonzero eigenvalues. Therefore we have
  \begin{align}
  {\bm{\lambda}}(\bm{\Pi}^{\frac{1}{2}}\bm{F}_1\bm{F}_1^{\rm H}\bm{\Pi}^{\frac{1}{2}})\hspace{-1mm}\succeq \hspace{-1mm}{\bm{\lambda}}(\bm{\Pi}^{\frac{1}{2}}\bm{F}_{\text{Pareto}}\bm{F}_{\text{Pareto}}^{\rm H}\bm{\Pi}^{\frac{1}{2}}),
  \end{align}based on which it can be concluded that we can find a matrix $\bm{F}_2$ which satisfies
\begin{align}
& \bm{F}_2\bm{F}_2^{\rm H} \preceq \bm{F}_1\bm{F}_1^{\rm H}, \nonumber \\
& {\bm{\lambda}}(\bm{\Pi}^{\frac{1}{2}}\bm{F}_2\bm{F}_2^{\rm H}\bm{\Pi}^{\frac{1}{2}})\hspace{-1mm}= \hspace{-1mm}{\bm{\lambda}}(\bm{\Pi}^{\frac{1}{2}}\bm{F}_{\text{Pareto}}\bm{F}_{\text{Pareto}}^{\rm H}\bm{\Pi}^{\frac{1}{2}}).
\end{align} As the weighted power constraint is right unitarily invariant, there will exist a unitary matrix $\bm{Q}_2$ making the following equality hold
\begin{align}
\bm{Q}_2^{\rm{H}}\bm{F}_2^{\rm H}\bm{\Pi}\bm{F}_2\bm{Q}_2\hspace{-1mm}= \hspace{-1mm}\bm{F}_{\text{Pareto}}^{\rm H}\bm{\Pi}\bm{F}_{\text{Pareto}}.
\end{align} Taking $\bm{F}_2\bm{Q}_2=\bm{F}_3$ as a new variable, it is concluded that for (\ref{ap4})
$\bm{F}_3^{\rm H}\bm{\Pi}\bm{F}_3\hspace{-1mm}= \hspace{-1mm}\bm{F}_{\text{Pareto}}^{\rm H}\bm{\Pi}\bm{F}_{\text{Pareto}}$ and  ${\sum}_{i=1}^I\hspace*{-2mm} \alpha_i {\rm Tr}\big(\bm{\Omega}_i\bm{F}_3\bm{F}_3^{\rm H}\big)<{\sum}_{i=1}^I\hspace*{-2mm} \alpha_i {\rm Tr}\big(\bm{\Omega}_i\bm{F}_2\bm{F}_2^{\rm H}\big)\le P$. This contradicts with previous conclusion. In other words, $\bm{F}_{\text{Pareto}}$ is a Pareto optimal solution of (\ref{label_opt}).

 In a nutshell, for any Pareto optimal solution of \textbf{Opt.\,1.6}, there exist the
 weights $\alpha_i$, $1\le i\le I$, for ensuring that this Pareto optimal solution of
 \textbf{Opt.\,1.6} is also the Pareto optimal solution of (\ref{label_opt}). Therefore,
 it can be concluded that any Pareto optimal solution of \textbf{Opt.\,1.6} satisfies
 the common structures of the Pareto optimal solutions of (\ref{label_opt}).

 Next, we show that the Pareto optimal solutions of (\ref{label_opt}) own the same
 diagonalizable structure and thus this structure is also the optimal structure of the
 Pareto optimal solutions of \textbf{Opt.\,1.6}. First define the
 auxiliary variables
\begin{align}\label{eq35}
 \widetilde{\bm{F}} =& \Big( \sum\nolimits_{i=1}^I \alpha_i\bm{\Omega}_i\Big)^{\frac{1}{2}}\bm{F}, \ \ \bm{\Omega}= \Big( \sum\nolimits_{i=1}^I \alpha_i\bm{\Omega}_i\Big)
\end{align} and the optimization (\ref{label_opt}) is transferred
 into:
\begin{align}\label{MMOP_A} 
 \max_{\widetilde{\bm{F}}} \ \ & \widetilde{\bm{F}}^{\rm H}\big(\bm{\Omega}^{-\frac{1}{2}}
  \big)^{\rm H} \bm{\Pi}\bm{\Omega}^{-\frac{1}{2}}\widetilde{\bm{F}},  \ \
 {\rm{s.t.}} \ {\rm Tr}\big(\widetilde{\bm{F}}\widetilde{\bm{F}}^{\rm H}\big) \le  P .
\end{align}
The Pareto optimal solution set of (\ref{MMOP_A}) consists of the optimal solutions of the
 following optimization problem for all the possible $\widetilde{\bm{F}}_{\rm in}$ that are
 in the sphere region of ${\rm Tr}\big(\widetilde{\bm{F}}\widetilde{\bm{F}}^{\rm H}\big)\le P$:
\begin{align}\label{ap7}
\hspace*{-4mm}\begin{array}{cl}\!\!
 \max\limits_{\widetilde{\bm{F}},\alpha} \!\!\! &\!\!\! \ \ \ \ \alpha , \\
 {\rm s.t.} \!\!\! &\!\!\! \widetilde{\bm{F}}^{\rm H} \big( \bm{\Omega}^{-\frac{1}{2}} \big)^{\rm H} \bm{\Pi}
  \bm{\Omega}^{-\frac{1}{2}} \widetilde{\bm{F}} \! =\! \alpha \widetilde{\bm{F}}^{\rm H}_{\rm in}
  \big( \bm{\Omega}^{-\frac{1}{2}} \big)^{\rm H} \bm{\Pi} \bm{\Omega}^{-\frac{1}{2}}
  \widetilde{\bm{F}}_{\rm in} , \\
 \!\!\! &\!\!\!  {\rm Tr}\big( \widetilde{\bm{F}} \widetilde{\bm{F}}^{\rm H} \big) \le P .
\!\! \end{array} \!\!
\end{align}
 The first constraint in (\ref{ap7}) is equivalent to
\begin{align}\label{ap8}
 \bm{\Pi}^{\frac{1}{2}}\bm{\Omega}^{-\frac{1}{2}}\widetilde{\bm{F}} = \sqrt{\alpha}\bm{U}
  \bm{\Pi}\bm{\Omega}^{-\frac{1}{2}}\widetilde{\bm{F}}_{\rm in} .
\end{align}
 Using pseudo inverse, we have
\begin{align}\label{ap9}
 \hspace*{-1mm}\big(\bm{\Pi}^{\frac{1}{2}}\bm{\Omega}^{-\frac{1}{2}}\big)^{\dag}\bm{\Pi}^{\frac{1}{2}}
  \bm{\Omega}^{-\frac{1}{2}}\widetilde{\bm{F}}\! =\!\! \sqrt{\alpha}\big(\bm{\Pi}^{\frac{1}{2}}
  \bm{\Omega}^{-\frac{1}{2}}\big)^{\dag}\bm{U}\bm{\Pi}\bm{\Omega}^{-\frac{1}{2}}\widetilde{\bm{F}}_{\rm in} ,\!
\end{align}based on which we have
\begin{align}
&\big\|\big(\bm{\Pi}^{\frac{1}{2}}\bm{\Omega}^{-\frac{1}{2}}\big)^{\dag}
  \bm{\Pi}^{\frac{1}{2}}\bm{\Omega}^{-\frac{1}{2}}\widetilde{\bm{F}}\big\|_{\rm{F}}^2 \nonumber \\
  =&{\rm{Tr}}\left( [\big(\bm{\Pi}^{\frac{1}{2}}\bm{\Omega}^{-\frac{1}{2}}\big)^{\dag}
  \bm{\Pi}^{\frac{1}{2}}\bm{\Omega}^{-\frac{1}{2}}\widetilde{\bm{F}}]^{\rm{H}}\big(\bm{\Pi}^{\frac{1}{2}}\bm{\Omega}^{-\frac{1}{2}}\big)^{\dag}
  \bm{\Pi}^{\frac{1}{2}}\bm{\Omega}^{-\frac{1}{2}}\widetilde{\bm{F}} \right) \nonumber \\
  =& \alpha{\rm{Tr}} \left([\big(\bm{\Pi}^{\frac{1}{2}}
  \bm{\Omega}^{-\frac{1}{2}}\big)^{\dag}\bm{U}\bm{\Pi}\bm{\Omega}^{-\frac{1}{2}}\widetilde{\bm{F}}_{\rm in}]^{\rm{H}}              \big(\bm{\Pi}^{\frac{1}{2}}
  \bm{\Omega}^{-\frac{1}{2}}\big)^{\dag}\bm{U}\bm{\Pi}\bm{\Omega}^{-\frac{1}{2}}\widetilde{\bm{F}}_{\rm in}   \right)\nonumber \\
  =& \alpha \big\|\big(\bm{\Pi}^{\frac{1}{2}}\bm{\Omega}^{-\frac{1}{2}}\big)^{\dag}\bm{U}\bm{\Pi}
  \bm{\Omega}^{-\frac{1}{2}}\widetilde{\bm{F}}_{\rm in}\big\|_{\rm{F}}^2.
\end{align}Therefore, $\alpha$ is solved as
\begin{align}\label{ap10}
 \alpha =& \frac{\big\|\big(\bm{\Pi}^{\frac{1}{2}}\bm{\Omega}^{-\frac{1}{2}}\big)^{\dag}
  \bm{\Pi}^{\frac{1}{2}}\bm{\Omega}^{-\frac{1}{2}}\widetilde{\bm{F}}\big\|_{\rm{F}}^2}
  {\big\|\big(\bm{\Pi}^{\frac{1}{2}}\bm{\Omega}^{-\frac{1}{2}}\big)^{\dag}\bm{U}\bm{\Pi}
  \bm{\Omega}^{-\frac{1}{2}}\widetilde{\bm{F}}_{\rm in}\big\|_{\rm{F}}^2} .
\end{align}
 Based on \textbf{Matrix Inequality 1}, the numerator of (\ref{ap10}) satisfies
\begin{align}\label{ap11}
 \big\|\big(\bm{\Pi}^{\frac{1}{2}}\bm{\Omega}^{-\frac{1}{2}}\big)^{\dag}\bm{\Pi}^{\frac{1}{2}}
  \bm{\Omega}^{-\frac{1}{2}}\widetilde{\bm{F}}\big\|_{\rm{F}}^2\le \sum\nolimits_j\lambda_j\big(
  \widetilde{\bm{F}}\widetilde{\bm{F}}^{\rm H}\big) ,
\end{align}
 while its denominator satisfies
\begin{align}\label{ap12}
 \big\|\big(\bm{\Pi}^{\frac{1}{2}}\bm{\Omega}^{-\frac{1}{2}}\big)^{\dag}\bm{U}\bm{\Pi}
  \bm{\Omega}^{-\frac{1}{2}}\widetilde{\bm{F}}_{\rm in}\big\|_{\rm{F}}^2 \ge \sum_j
   \frac{\sigma^2_j\big(\bm{\Pi}\bm{\Omega}^{-\frac{1}{2}}\widetilde{\bm{F}}_{\rm in}\big)}
  {\sigma^2_j(\bm{\Pi}^{\frac{1}{2}}\bm{\Omega}^{-\frac{1}{2}}\big)} ,
\end{align}
 where $\sigma_j^2(\bm{A})$ denotes the  $j$th singular value of $\bm{A}$. Clearly, $\alpha$
 attains the maximum value when the both equalities in (\ref{ap11}) and (\ref{ap12}) hold. For the optimal $\widetilde{\bm{F}}$ and $\bm{U}$ together with the fact that for
 \textbf{Opt.\,1.6}, the optimal $\bm{F}$ is right unitary invariant, the optimal $\widetilde{\bm{F}}$ satisfies the following structure
 \begin{align}
 \label{F_appendix}
 \widetilde{\bm{F}} =& \bm{U}_{\widetilde{\bm{\Pi}}}
  \bm{\Lambda}_{\widetilde{\bm{F}}}\bm{U}_{\rm Arb}^{\rm H},
\end{align} where the unitary matrix $\bm{U}_{\widetilde{\bm{\Pi}}}$ is defined based on the following EVD
\begin{align}
 \bm{\Omega}^{-\frac{1}{2}}\bm{\Pi}\bm{\Omega}^{-\frac{1}{2}} =& \bm{U}_{\widetilde{\bm{\Pi}}}
  \bm{\Lambda}_{\widetilde{\bm{\Pi}}}\bm{U}_{\widetilde{\bm{\Pi}}}^{\rm H} \ \text{with} \
  \bm{\Lambda}_{\widetilde{\bm{\Pi}}} \searrow .
\end{align}Based on (\ref{F_appendix}) and  the definition of  $\widetilde{\bm{F}}$ in (\ref{eq35})
\begin{align}
{\bm{F}} =& \bm{\Omega}^{-\frac{1}{2}}\bm{U}_{\widetilde{\bm{\Pi}}}
  \bm{\Lambda}_{\widetilde{\bm{F}}}\bm{U}_{\rm Arb}^{\rm H}.
\end{align}


\section{Bayes Robust Matrix-Monotonic Optimization }
\label{appendix_3}

\subsection{Shaping Constraint}

With  the shaping constraint, \textbf{Opt.\,2.2} becomes
\begin{align}\label{opt_appendix_1}
\begin{array}{lcl}
 & \max\limits_{\bm{F}} & \bm{F}^{\rm H}\widehat{\bm{H}}^{\rm H}
  \bm{K}_{\rm n}^{-1}\widehat{\bm{H}}\bm{F} , \\
 & {\rm{s.t.}} & \bm{K}_{\rm n}=\sigma_{\rm{n}}^2\bm{I}+{\rm Tr}\big(\bm{F}\bm{F}^{\rm H}
  \bm{\Psi}\big)\bm{I} , \bm{F}\bm{F}^{\rm{H}} \preceq {\bm{R}}_{\rm{s}}.
\end{array}
\end{align} Note that ${\rm Tr}\big(\bm{F}\bm{F}^{\rm H}
  \bm{\Psi}\big)\le {\rm{Tr}}(\bm{R}_{\rm{s}}\bm{\Psi})$ and then we have the following matrix inequality
\begin{align}
\label{app_bower_bound}
 \bm{F}^{\rm H}\widehat{\bm{H}}^{\rm H}
  \bm{K}_{\rm n}^{-1}\widehat{\bm{H}}\bm{F} \succeq \frac{\bm{F}^{\rm H}\widehat{\bm{H}}^{\rm H}
  \widehat{\bm{H}}\bm{F}}{\sigma_{\rm{n}}^2+{\rm{Tr}}(\bm{R}_{\rm{s}}\bm{\Psi})}.
\end{align} Replacing the objective in (\ref{opt_appendix_1}) by its lower bound in (\ref{app_bower_bound}), the following optimization problem is achieved
\begin{align}\label{opt_appendix_1_1}
\begin{array}{lcl}
 & \max\limits_{\bm{F}} & \frac{\bm{F}^{\rm H}\widehat{\bm{H}}^{\rm H}
  \widehat{\bm{H}}\bm{F}}{\sigma_{\rm{n}}^2+{\rm{Tr}}(\bm{R}_{\rm{s}}\bm{\Psi})} ,  {\rm{s.t.}}  \bm{F}\bm{F}^{\rm{H}} \preceq {\bm{R}}_{\rm{s}}.
\end{array}
\end{align} whose Pareto optimal solution is given by Lemma~\ref{L1}. It is obvious that when $\bm{\Psi}=0$ or $\bm{\Psi}\propto \bm{I}$ and ${\rm Tr}\big(\bm{F}\bm{F}^{\rm H})={\rm{Tr}}(\bm{R}_{\rm{s}})$ is achievable the lower bound is tight.

\subsection{Joint Power Constraints}

Under the joint power constraints, \textbf{Opt.\,2.2} is written in the following formula
\begin{align}\label{opt_appendix_2}
\begin{array}{lcl}
 & \max\limits_{\bm{F}} & \bm{F}^{\rm H}\widehat{\bm{H}}^{\rm H}
  \bm{K}_{\rm n}^{-1}\widehat{\bm{H}}\bm{F} , \\
 & {\rm{s.t.}} & \bm{K}_{\rm n}=\sigma_{\rm{n}}^2\bm{I}+{\rm Tr}\big(\bm{F}\bm{F}^{\rm H}
  \bm{\Psi}\big)\bm{I} ,  \\
 & & {\rm{Tr}}\big(\bm{F}\bm{F}^{\rm{H}}\big) \le P, \ \bm{F}\bm{F}^{\rm{H}} \preceq \tau {\bm{I}}.
\end{array}
\end{align}
The sum power constraint ${\rm Tr}\big(\bm{F}\bm{F}^{\rm H}\big) \le P$ is equivalent to the following equality \cite{XingTSP201501}
\begin{align}
\label{Inequality_App}
 {\left(\sigma_{\rm{n}}^2+{\rm Tr}\big(\bm{F}\bm{F}^{\rm H}\bm{\Psi}\big)\right)^{-1}} {{\rm Tr}\big[\big(\sigma_{\rm{n}}^2\bm{I}+P\bm{\Psi}\big)\bm{F}\bm{F}^{\rm H}\big]}
 \le P.
\end{align} Based on (\ref{Inequality_App}), the optimization problem (\ref{opt_appendix_2}) is equivalent to the following one
\begin{align}\label{opt_appendix_3}
\hspace*{-2mm}\begin{array}{lcl}
 \!\! &\!\! \max\limits_{\bm{F}}\!\! &\!\! \bm{F}^{\rm H}\widehat{\bm{H}}^{\rm H}
  \bm{K}_{\rm n}^{-1}\widehat{\bm{H}}\bm{F} , \\
 \!\! &\!\! {\rm{s.t.}}\!\! &\!\! \bm{K}_{\rm n}=\sigma_{\rm{n}}^2\bm{I}+{\rm Tr}\big(\bm{F}\bm{F}^{\rm H}
  \bm{\Psi}\big)\bm{I} ,  \\
 \!\! &\!\! \!\! &\!\! \frac{{\rm Tr}\big[\big(\sigma_{\rm{n}}^2\bm{I}+P\bm{\Psi}\big)\bm{F}\bm{F}^{\rm H}\big]}
  {\sigma_{\rm{n}}^2+{\rm Tr}\big(\bm{F}\bm{F}^{\rm H}\bm{\Psi}\big)}\le P, \ \bm{F}\bm{F}^{\rm H} \preceq \tau\bm{I}.
\end{array}\!\!
\end{align}

By defining the following matrix variable
\begin{align}
\label{F_tilde_app}
 \widetilde{\bm{F}} =& {\big[\sigma_{\rm{n}}^2+{\rm Tr}\big(\bm{F}\bm{F}^{\rm H}\bm{\Psi}\big)\big]^{-\frac{1}{2}}}\big(\sigma_{\rm{n}}^2\bm{I}+P\bm{\Psi}\big)^{\frac{1}{2}}
  \bm{F} ,
\end{align}the optimization problem (\ref{opt_appendix_3}) can be transferred into the following equivalent one
\begin{align}\label{opt_appendix_4}
\hspace*{-3mm}\begin{array}{lcl}
\!\! &\!\! \max\limits_{ \widetilde{\bm{F}}}\!\! &\!\!\!  {\widetilde{\bm{F}}}^{\rm H}\widehat{\bm{H}}^{\rm H}
  \widehat{\bm{H}} \widetilde{\bm{F}} , \\
 \!\! & \!\!{\rm{s.t.}}\!\! &\!\!\! {\rm Tr}\big({\widetilde{\bm{F}}}
 {\widetilde{\bm{F}}}^{\rm H}\big)\! \le\! P, \  {\widetilde{\bm{F}}} {\widetilde{\bm{F}}}^{\rm H} \preceq \tau \frac{\sigma_{\rm{n}}^2\bm{I}+P\bm{\Psi}}{\sigma_{\rm{n}}^2+{\rm Tr}\big(\bm{F}\bm{F}^{\rm H}\bm{\Psi}\big)}.
\end{array}
\end{align}
For the final matrix inequality, we have the following lower bound of the righthand side term, i.e.,
\begin{align}
\label{lower_bound}
 \tau\frac{\sigma_{\rm{n}}^2+P\lambda_{\min}(\bm{\Psi})}{\sigma_{\rm{n}}^2+P\lambda_{\max}(\bm{\Psi})}\bm{I} \preceq \tau\frac{\sigma_{\rm{n}}^2\bm{I}+P\bm{\Psi}}{\sigma_{\rm{n}}^2+{\rm Tr}\big(\bm{F}\bm{F}^{\rm H}\bm{\Psi}\big)},
\end{align}where the equality holds when $\bm{\Psi} \propto \bm{I}$. Based on the lower bound in (\ref{lower_bound}), for the Pareto optimal solutions of the following optimization problem, the corresponding objective is a lower bound of that in (\ref{opt_appendix_3})
\begin{align}\label{opt_appendix_5}
\hspace*{-3mm}\begin{array}{lcl}
\!\! &\!\! \max\limits_{ \widetilde{\bm{F}}}\!\! &\!\!\!  {\widetilde{\bm{F}}}^{\rm H}\big(\sigma_{\rm{n}}^2\bm{I}+P\bm{\Psi}\big)^{-\frac{1}{2}}\widehat{\bm{H}}^{\rm H}
  \widehat{\bm{H}}\big(\sigma_{\rm{n}}^2\bm{I}+P\bm{\Psi}\big)^{-\frac{1}{2}} \widetilde{\bm{F}} , \\
 \!\! & \!\!{\rm{s.t.}}\!\! &\!\!\! {\rm Tr}\big({\widetilde{\bm{F}}}
 {\widetilde{\bm{F}}}^{\rm H}\big)\! \le\! P, \ {\widetilde{\bm{F}}} {\widetilde{\bm{F}}}^{\rm H} \preceq \tau\frac{\sigma_{\rm{n}}^2+P\lambda_{\max}(\bm{\Psi})}{\sigma_{\rm{n}}^2+P\lambda_{\min}(\bm{\Psi})}\bm{I}.
\end{array}
\end{align} It is obvious that based on Lemma~\ref{L2} the Pareto optimal solutions of (\ref{opt_appendix_5}) satisfy the following structure
\begin{align}
\widetilde{\bm{F}}={\bm{V}}_{\widetilde{\bm{H}}}\bm{\Lambda}_{\widetilde{\bm{F}}}{\bm{U}}_{\rm{Arb}}^{\rm{H}}
\end{align}where the unitary matrix ${\bm{V}}_{\widehat{\bm{H}}}$ is defined based on the SVD
\begin{align}
\widehat{\bm{H}}\big(\sigma_{\rm{n}}^2\bm{I}+P\bm{\Psi}\big)^{-\frac{1}{2}}=
   {\bm{U}}_{\widetilde{\bm{H}}}\bm{\Lambda}_{\widetilde{\bm{H}}}{\bm{V}}_{\widetilde{\bm{H}}}^H, \text{with}, \bm{\Lambda}_{\widetilde{\bm{H}}} \searrow .
 \end{align}  The diagonal elements of the rectangular diagonal matrix $\bm{\Lambda}_{\widetilde{\bm{F}}_k}$ are smaller than
$\sqrt{\tau{(\sigma_{\rm{n}}^2+P\lambda_{\min}(\bm{\Psi}))}/{(\sigma_{\rm{n}}^2+P\lambda_{\max}(\bm{\Psi}))}}$. Based on the definition in (\ref{F_tilde_app}), $\bm{F}$ equals
\begin{align}
\hspace*{-2mm}\bm{F}\hspace*{-1mm}=\hspace*{-1mm}\frac{\sigma_{\rm{n}}
\widetilde{\bm{\Psi}}^{-\frac{1}{2}}
{\bm{V}}_{\widetilde{\bm{H}}}\bm{\Lambda}_{\widetilde{\bm{F}}}{\bm{U}}_{\rm{Arb}}^{\rm{H}}}
{\left(1\hspace*{-1mm}-\hspace*{-1mm}
{\rm{Tr}}\left(\widetilde{\bm{\Psi}}^{-\frac{1}{2}}
\bm{\Psi}\widetilde{\bm{\Psi}}^{-\frac{1}{2}}{\bm{V}}_{\widetilde{\bm{H}}}
\bm{\Lambda}_{\widetilde{\bm{F}}}\bm{\Lambda}_{\widetilde{\bm{F}}}^{\rm{H}}
{\bm{V}}_{\widetilde{\bm{H}}}^{\rm{H}}\right)\right)^{\frac{1}{2}}},
\end{align} where $\widetilde{\bm{\Psi}}=\sigma_{\rm{n}}^2\bm{I}+P\bm{\Psi}$.


\begin{thebibliography}{99}

\bibitem{JYang1994} 
 J.~Yang and S.~Roy, ``On joint transmitter and receiver optimization for
 multiple-input-multiple-output (MIMO) transmission systems,'' {\em IEEE Trans. Commun.},
 vol.~42, no.~12, pp.~3221--3231, Dec. 1994.

\bibitem{Alamouti1998} 
 S.~M.~Alamouti, ``A simple transmit diversity technique for wireless communications,''
 {\em IEEE J. Sel. Areas Commun.}, vol.~16, no.~8, pp.~1451--1458, Oct. 1998.

\bibitem{Telatar1999} 
 I.~E.~Telatar, ``Capacity of multi-antenna Gaussian channels,'' {\em European Trans.
 Commun.}, vol.~10, no.~2, pp.~585--595, Nov./Dec. 1999.

\bibitem{Sampth01} 
 H.~Sampath, P.~Stoica, and A.~Paulraj, ``Generalized linear precoder and decoder design
 for MIMO channels using the weighted MMSE criterion,'' {\em IEEE Trans. Commun.}, vol.~49,
 no.~12, pp.~2198--2206, Dec. 2001.

\bibitem{Sampth03} 
 H.~Sampath and A.~Paulraj, ``Linear precoding for space-time coded systems with known
 fading correlations,'' {\em IEEE Commun. Lett.}, vol.~6, no.~6, pp.~239--241, Jun. 2002.

\bibitem{Scaglione2002} 
 A.~Scaglione, {\em et al.}, ``Optimal designs for space-time linear precoders and
 decoders,'' {\em IEEE Trans. Signal Proces.}, vol.~50, no.~5, pp.~1051--1064, May 2002.


\bibitem{Feiten2007} 
 A.~Feiten, R.~Mathar, and S.~Hanly, ``Eigenvalue-based optimum-power allocation for
 Gaussian vector channels,''{\em IEEE Trans. Inf. Theory}, vol.~53, no.~6,
 pp.~2304--2309, Jun. 2007.

\bibitem{Yadav2014} 
 A.~Yadav, M.~Juntti, and J.~Lilleberg, ``Linear precoder design for doubly correlated
 partially coherent fading MIMO channels,'' {\em IEEE Trans. Wireless Commun.}, vol.~13,
 no.~7, pp.~3621--3635, Jul. 2014.


\bibitem{Palomar03} 
 D.~P.~Palomar, J.~M.~Cioffi, and M.~A.~Lagunas, ``Joint Tx-Rx beamforming design for
 multicarrier MIMO channels: A unified framework for convex optimization,'' {\em IEEE Trans. Signal Process.}, vol.~51, no.~9, pp.~2381--2401, Sep. 2003.


\bibitem{Majorization} 
 D.~P.~Palomar and Y.~Jiang, ``MIMO transceiver designs via majorization theory,''
 {\em  Foundations and Trends in Commun. and Inf. Theory}, vol.~3, no.~4-5, pp~331--551,
 Jun. 2007.

\bibitem{Jiang2005} 
 Y.~Jiang, J.~Li, and W.~W.~Hager, ``Joint transceiver design for MIMO communications using
 geometric mean decomposition,'' {\em IEEE Trans. Signal Process.}, vol.~53, no.~10,
 pp.~3791--3803, Oct. 2005.


\bibitem{Weng2010TSP1}
C. Weng, C. Chen, and P. P. Vaidyanathan, ``MIMO transceivers with decision feedback and bit loading: Theory and optimization,'' {\em IEEE Trans. Signal Process.}, vol. 58, no. 3, pp. 1334--1346, Mar. 2010.


\bibitem{Liu2013TSP}
C. Liu and P. P. Vaidyanathan, ``MIMO broadcast DFE transceivers with QoS constraints: Min-power and max-rate solutions,'' {\em IEEE Trans. Signal Process.}, vol. 61, no. 22, pp. 5550--5562, Nov. 2013.


\bibitem{Weng2010TSP2}
C. Weng and P. P. Vaidyanathan, ``MIMO transceiver optimization with linear constraints on transmitted signal covariance components,'' {\em IEEE Trans. Signal Process.,} vol. 58, no. 1, pp. 458--462, Jan. 2010.






\bibitem{Jafar2005} 
 S.~A.~Jafar and A.~Goldsmith, ``Multiple-antenna capacity in correlated Rayleigh fading
 with channel covariance information,'' {\em IEEE Trans. Wireless Commun.,} vol.~4, no.~3,
 pp.~990--997, May 2005.

\bibitem{Jafar2004} 
 S.~A.~Jafar and A.~Goldsmith, ``Transmitter optimization and optimality of beamforming for multiple antenna systems,'' {\em IEEE Trans. Wireless Commun.,} vol.~3, no.~4,
 pp.~1165--1175, Jul. 2004.

\bibitem{Zhang2008} 
 X.~Zhang, D.~P.~Palomar, and B.~Ottersten, ``Statistically robust design of linear MIMO
 transceivers,'' {\em IEEE Trans. Signal Process.}, vol.~56, no.~8, pp.~3678--3689,
 Aug. 2008.

\bibitem{Ding09} 
 M.~Ding and S.~D.~Blostein, ``MIMO minimum total MSE transceiver design with imperfect
 CSI at both ends,'' {\em IEEE Trans. Signal Process.}, vol.~57, no.~3, pp.~1141--1150,
 Mar. 2009.

\bibitem{JHWang2013} 
 J.~Wang, M.~Bengtsson, B.~Ottersten, and D.~P.~Palomar, ``Robust MIMO precoding for
 several classes of channel uncertainty,'' {\em IEEE Trans. Signal Process.}, vol.~61,
 no.~12, pp.~3056--3070, Jun. 2013.

\bibitem{Pastore2012} 
 A.~Pastore, M.~Joham, and J.~R.~Fonollosa, ``A framework for joint design of pilot
 sequence and linear precoder,'' {\em IEEE Trans. Inf. Theory}, vol.~62, no.~9,
 pp.~5059--5079, Sep. 2016.




\bibitem{WYu2007} 
 W.~Yu and T.~Lan, ``Transceiver optimization for the multi-antenna downlink with
 per-antenna power constraints,'' {\em IEEE Trans. Signal Process.}, vol.~55, no.~6,
 pp.~2646--2660, Jun. 2007.



\bibitem{JFang2013} 
 J.~Fang, H.~Li, Z.~Chen, and Y.~Gong, ``Joint precoder design for distributed
 transmission of correlated sources in sensor networks,'' {\em IEEE Trans. Wireless Commun.,}
 vol.~12, no.~6, pp.~2918--2929, Jun. 2013.

\bibitem{XingIET} 
 C.~Xing, S.~Li, Z.~Fei, and J.~Kuang, ``How to understand linear minimum mean square
 error transceiver design for multiple input multiple output systems from quadratic
 matrix programming,'' {\em IET Commun.}, vol.~7, no.~12, pp.~1231--1242, Aug. 2013.



\bibitem{XingTSP2013} 
 C.~Xing, {\em et al.}, ``A general robust linear transceiver design for
 amplify-and-forward multi-hop MIMO relaying systems,'' {\em IEEE Trans. Signal Process.},
 vol.~61, no.~5, pp.~1196--1209, Mar. 2013.

\bibitem{XingJSAC2012} 
 C.~Xing, M.~Xia, F.~Gao and Y.-C.~Wu, ``Robust transceiver with Tomlinson-Harashima
 precoding for amplify-and-forward MIMO relaying systems,'' {\em IEEE J. Sel. Areas Commun.},
 vol.~30, no.~8, pp.~1370--1382, Sep. 2012.


\bibitem{XingTSP201501} 
 C.~Xing, S.~Ma, and Y.~Zhou, ``Matrix-monotonic optimization for MIMO systems,''
 {\em IEEE Trans. Signal Process.}, vol.~63, no.~2, pp.~334--348, Jan. 2015.


\bibitem{Jorswieck07} 
 E.~Jorswieck and H.~Boche, ``Majorization and matrix-monotone functions in wireless
 communications,'' {\em  Foundations and Trends in Commun. and Inf. Theory}, vol.~3, no.~6, pp~553--701, Jul. 2007.




\bibitem{Wu2016}
S. X. Wu, Q. Li, A. M. So, and W. Ma, ``A stochastic beamformed amplify-and-forward scheme in a multigroup multicast MIMO relay network with per-antenna power constraints,'' {\em IEEE Trans. Wireless Commun.,} vol. 15, no. 7, pp. 4973--4986, Jul. 2016.


\bibitem{Toli2008}
A. T${\rm{\ddot{o}}}$lli, M. Codreanu, and M. Juntti, ``Linear multiuser MIMO transceiver design with quality of service and per-antenna power constraints,'' {\em IEEE Trans. Signal Process.}, vol. 56, no. 7, pp. 3049--3055, Jul. 2008.

\bibitem{Dong2013}
M. Dong, B. Liang, and Q. Xiao, ``Unicast multi-antenna relay beamforming with per-antenna power control: Optimization and duality,'' {\em IEEE Trans. Signal Process.}, vol. 61, no. 23, pp. 6076--6090, Dec. 2013.

\bibitem{Christopoulos}
D. Christopoulos, S. Chatzinotas, and B. Ottersten, ``Weighted fair multicast multigroup beamforming under per-antenna power constraints,'' {\em IEEE Trans. Signal Process.}, vol. 62, no. 19, pp. 5132--5142, Oct. 2014.


\bibitem{Luo2008}
V. Havary-Nassab, S. Shahbazpanahi, A. Grami, and Z.-Q. Luo, ``Distributed beamforming for relay networks
based on second-order statistics of the channel state information,'' {\em IEEE Signal Process.,} vol.~56, no.~9,
pp.~4306--4316, Sep. 2008.





\bibitem{Mai2011} 
 M.~Vu, ``MIMO capacity with per-antenna power constraint,'' in {\em Proc. GLOBECOM 2011}
 (Houston, USA), Dec. 5-9, 2011, pp.~1--5.

\bibitem{XingTSP201601} 
 C.~Xing, Y.~Ma, Y.~Zhou, and F.~Gao, ``Transceiver optimization for multi-hop
 communications with per-antenna power constraints,''  {\em IEEE Trans. Signal Process.},
 vol.~64, no.~6, pp.~1519--1534, Mar. 2016.

\bibitem{Palomar2004}
D. P. Palomar, ``Unified framework for linear MIMO transceivers with shaping constraints,'' {\em IEEE Communi. Lett.,} vol. 8, no. 12, pp. 697--699, Dec. 2004.

\bibitem{XingTSP201502} 
 C.~Xing, F.~Gao, and Y.~Zhou, ``A framework for transceiver designs for multi-hop
 communications with covariance shaping constraints,'' {\em IEEE Trans. Signal Process.},
 vol.~63, no.~15, pp.~3930--3945, Aug. 2015.

\bibitem{Dai2012} 
 J.~Dai, C.~Chang, W.~Xu, and Z.~Ye, ``Linear precoder optimization for MIMO systems
 with joint power constraints,'' {\em IEEE Trans. Commun.}, vol.~60, no.~8, pp.~2240--2254,
 Aug. 2012.


\bibitem{NewTSP} 
 S.~Wang, S.~Ma, C.~Xing, S.~Gong, and J.~An, ``Optimal training design for MIMO systems with general power
 constraints,'' {\em Trans. Signal Process.}, vol.~66, no.~14, pp.~3649--3664, Jul. 2018.



\bibitem{XingCL}
C.~Xing, W.~Li, S.~Ma, Z.~Fei, and J.~Kuang ``A matrix-field weighted mean-square-error model for MIMO transceiver design,'' {\em Commun. Letter}, vol. 17, no.~8, pp.~1652--1655, Aug. 2013.

\bibitem{Horn1990} 
 R.~A.~Horn and C.~R.~Johnson, {\em Matrix Analysis}. Cambridge University Press: Cambridge, UK,
 1990.


\bibitem{BarrySimon}
B.~Simon, {\em Loewner's Theorem on Monotone Matrix Functions (Grundlehren der mathematischen Wissenschaften Book 354).} Switzerland: Springer, 2019.

\bibitem{Xingzhi_Zhang}
X.~Zhang, {\em Matrix Inequalities  (Lecture Notes in Mathematics).} New York, NY, USA: Springer,  2002.




\bibitem{Marshall79} 
 A.~W.~Marshall and I.~Olkin, {\em Inequalities: Theory of Majorization and Its
 Applications}. New York: Academic Press, 1979.


%
%
%
%





\bibitem{XingTVT2016} 
 C.~Xing, Y.~Jing, and Y.~Zhou, ``On weighted MSE model for MIMO transceiver optimization,''
 {\em IEEE Trans. Veh. Techno.}, vol.~66, no.~8, pp.~7072--7085, Aug. 2017.

\bibitem{Boyd04}  
 S.~Boyd and L.~Vandenberghe, {\em Convex Optimization}. Cambridge University Press: Cambridge,
 UK, 2004.


\bibitem{General_Waterfilling}
C. Xing, Y. Jing, S. Wang, S. Ma, and H. V. Poor, ``New viewpoint and algorithms for water-filling solutions in wireless communications,'' {\em IEEE Trans. Signal Process.,} vol.~68, pp.~1618--1634, Feb. 12, 2020.


\bibitem{ShiqiUAV}
S. Gong, S. Wang, C. Xing, S. Ma, and T. Q. S. Quek, ``Robust superimposed training optimization for UAV assisted communication systems,'' {\em IEEE Trans. Commun.,} vo. 19, no. 3, pp. 1704--1721, March 2020.



\bibitem{tool} 
 M.~C.~Grant and S.~P.~Boyd, \emph{The CVX Users' Guide} (Release 2.1) CVX Research, Inc.,
 2015

\end{thebibliography}
\end{document}